%% file: main.tex
\documentclass[submission,copyright,creativecommons]{eptcs}
\providecommand{\event}{LTT2026} 

\usepackage{iftex}

\ifpdf
  \usepackage{underscore}         
  \usepackage[T1]{fontenc}        
\else
  \usepackage{breakurl}           
\fi

\title{Principal Typing for Intersection Types, \\ Forty-Five Years Later}

\author{
Daniele Pautasso 
\institute{University of Turin}
\email{daniele.pautasso@unito.it}
\and 
Simona Ronchi Della Rocca
\institute{University of Turin}
\email{ronchi@di.unito.it}
}

\def\titlerunning{Principal Typing for Intersection Types, Forty-Five Years Later}

\def\authorrunning{D. Pautasso \& S. Ronchi Della Rocca}

\usepackage{amsthm}
\usepackage{amsmath}
\usepackage{amssymb}
\usepackage{enumerate}
\usepackage[utf8]{inputenc}
\usepackage{proof}
\usepackage{bussproofs}
\usepackage{algorithm}
\usepackage{algpseudocode} 
\usepackage{cleveref}
\usepackage{stmaryrd}
\usepackage{xspace}
\usepackage{color}
\usepackage{lineno}
\usepackage{hyperref}
\usepackage{cancel}

\theoremstyle{plain}
\newtheorem{theorem}{Theorem}
\newtheorem{lemma}[theorem]{Lemma}
\newtheorem*{theorem*}{Theorem}
\newtheorem*{lemma*}{Lemma}
\newtheorem{property}[theorem]{Property}
\newtheorem{corollary}[theorem]{Corollary}

\theoremstyle{definition}
\newtheorem{definition}[theorem]{Definition}
\newtheorem{notation}[theorem]{Notation}
\newtheorem{example}[theorem]{Example}

\theoremstyle{remark}
\newtheorem{remark}{Remark}

\input{macros}

\begin{document}

\maketitle

\begin{center}
\textit{Dedicated to Stefano Berardi, \\who brought a clear and illuminating logical light to computational problems.}
\end{center}

\begin{abstract}
A type assignment system for $\lambda$-calculus enjoys the principal
typing property if every typable term $M$ has a special typing, called principal, 
from which all typings for $M$ can be obtained via
suitable operations. The existence of principal typings 
in various intersection type disciplines has long been established 
using both semantical and syntactical approaches.
Historically, on the syntactical side, 
proofs of this property and the description 
of type inference (semi-)algorithms computing principal typings
have been complicated by many subtle technicalities;
the present work aims at providing a more accessible formulation. 
To this end, we revisit some classical notions, and  
identify three elementary operations (substitution, expansion, erasure) 
that allow to build any type derivation in a system
characterizing head normalization.
We then use such operations 
in the design of an inference semi-algorithm
that computes the principal typing of all and only the strongly normalizing terms,
thus contributing to a modern perspective on results originally proven
more than 40 years ago.
\end{abstract}


\section{Introduction}
\input{intro}

\section{Preliminaries} \label{sec:preliminaries}
\input{preliminaries}

\section{Towards Intersection Type Inference} \label{sec:intersection-types}
\input{intersection}

\section{A Type Inference Semi-algorithm} \label{sec:algorithm}
\input{algo}

\section{Concluding Remarks} \label{sec:conclusion}
\input{conclusion}

\bibliographystyle{eptcs}
\bibliography{biblio}

\appendix
\input{appendix}

\end{document}

%% file: macros.tex
\renewcommand{\l}{\lambda}

\newcommand{\outvar}{\mathsf{OutVar}}
\newcommand{\varIT}{\mathbf{V}}
\newcommand{\varPT}{\mathbf{V}_\mathtt{p}}
\newcommand{\FV}{\textsf{FV}}

\newcommand{\redb}{\rightarrow_{\beta}}
\newcommand{\redlo}{\rightarrow_{\mathsf{LO}}}

\newcommand{\ir}{\mathsf{I}}
\newcommand{\kr}{\mathsf{K}}

\newcommand{\con}{{\tt a}}
\newcommand{\conb}{{\tt b}}
\newcommand{\conc}{{\tt c}}
\newcommand{\cond}{{\tt d}}
\newcommand{\cone}{{\tt e}}
\newcommand{\conf}{{\tt f}}
\renewcommand{\cong}{{\tt g}}
\newcommand{\conh}{{\tt h}}
\newcommand{\coni}{{\tt i}}
\newcommand{\conl}{{\tt l}}

\newcommand{\A}{{\tt A}}
\newcommand{\B}{{\tt B}}

\newcommand{\Tvar}{\mathsf{Var}}

\newcommand{\iA}{A}
\newcommand{\iB}{B}

\newcommand{\arrow}{\rightarrow}

\newcommand{\der}{\vdash}
\newcommand{\dersn}{\vdash_{\snsym}}

\newcommand{\derp}{\vdash_{\tt p}}

\newcommand{\IT}{\mathbf{T}}
\newcommand{\ITsn}{\mathbf{T}_{\snsym}}
\newcommand{\PT}{\mathbf{T}_{\mathtt{p}}}

\newcommand{\dem}{\, \triangleright \,}

\newcommand{\ie}{\textit{i.e.}\xspace}

\newcommand{\ih}{\textit{i.h.}\xspace}

\newcommand{\mult}[1]{[#1]}
\newcommand{\emult}{\mult{\,}}
\newcommand{\llist}[1]{\langle #1 \rangle}
\newcommand{\ellist}{\langle \; \rangle}
\newcommand{\set}[1]{\{ #1 \}}

\newcommand{\subs}[2]{[#2 / #1]}

\newcommand{\dom}[1]{\mathsf{dom}(#1)}

\newcommand{\cod}[1]{\mathsf{cod}(#1)}

\newcommand{\ccontext}{{\tt C}}
\newcommand{\kcontext}{{\tt K}}

\newcommand{\axvar}{\mathsf{var}}
\newcommand{\abs}{\mathsf{abs}}
\newcommand{\absi}{\mathsf{abs}\text{-}\ir}
\newcommand{\absk}{\mathsf{abs}\text{-}\kr}
\newcommand{\app}{\mathsf{app}}
\newcommand{\many}{\mathsf{many}}

\newcommand{\InferStrong}{{\tt InferStrong}}

\newcommand{\snsym}{\mathsf{s}}

\newcommand{\Blocked}{{\tt Blocked}}
\newcommand{\PDmin}[1]{{\tt PD^{min}}(#1)}
\newcommand{\PDsnmin}[1]{{\tt PD^{min}_{\snsym}}(#1)}
\newcommand{\PD}[1]{{\tt PD}(#1)}
\newcommand{\PDsn}[1]{{\tt PD}_{\snsym}(#1)}

\newcommand{\Expand}[1]{{\tt Expand}(#1)}
\newcommand{\Erase}[1]{{\tt Erase}(#1)}

\newcommand{\systemIT}{{\mathcal N}}
\newcommand{\systemITsn}{\systemIT_{\snsym}}

\newcommand{\uniseq}{\bar{\mathsf{u}}}

\newcommand{\rewu}{\to_\mathsf{u}}

\newcommand{\rewur}{\rewu^*}
\newcommand{\nf}[1]{\mathtt{nf}_\mathsf{u}(#1)}
\newcommand{\nfo}[1]{\mathtt{nf}_\mathsf{o}(#1)}
\newcommand{\rewo}{\to_\mathsf{o}}

\newcommand{\rewor}{\rewo^*}
\newcommand{\rewe}{\to_\mathsf{e}}

\newcommand{\Rew}{\Rightarrow}

\newcommand{\ms}{\mu}
\newcommand{\mstwo}{\nu}

\newcommand{\ls}{\sigma}
\newcommand{\lstwo}{\tau}
\newcommand{\cons}{\boldsymbol{\cdot}}

\newcommand{\gpipe}{ \ \mid \ }
\newcommand{\grameq}{ \ ::= \ }

\newcommand{\cminus}{\setminus}

\newcommand{\eqset}{E}

\newcommand{\ltom}{\mathsf{m}}

\newcommand{\mgu}[1]{\vec{#1}}

\newcommand{\expseq}{\bar{\mathsf{e}}}
\newcommand{\expstep}{\mathsf{e}}

\newcommand{\uerase}{\textit{erase}}
\newcommand{\uswap}{\textit{swap}}
\newcommand{\uarrow}{\textit{arrow}}
\newcommand{\ulist}{\textit{list}}
\newcommand{\usubszero}{\textit{subs-out}}
\newcommand{\usubs}{\textit{subs}}

\newcommand{\zerosubs}[2]{\llparenthesis #2/#1 \rrparenthesis}

\newcommand{\tDelta}{\mathtt{\Delta}}
\newcommand{\tGamma}{\mathtt{\Gamma}}
\newcommand{\tPhi}{\mathtt{\Phi}}
\newcommand{\tPsi}{\mathtt{\Upsilon}}

\newcommand{\eabs}{\text{abs}}
\newcommand{\eappl}{\text{appl}}
\newcommand{\eappr}{\text{appr}}

%% file: intro.tex

The notion of \emph{principal type} was introduced by Hindley in \cite{Hindley69}, 
and further developed in \cite{Hindley97}, in the setting of simple type assignment 
for $\lambda$-calculus\footnote{Actually, to be more precise, for combinatory logic.}. 
The principal type of a $\lambda$-term $M$ 
is the most general type derivable for it, 
in the sense that all types one can assign to $M$ are obtained 
from the principal one by means of substitution.
Hindley proved that 
 all simply typable terms have a principal type, 
and moreover that the (principal) typability problem is decidable: 
in fact, simple type inference can be regarded as an instance of
Robinson's classical unification problem \cite{Robinson65}. 
Ben-Yelles used these ideas in the design of his 
principal type inference algorithm \cite{Ben-Yelles79}, 
which later provided the basis for the type assignment 
procedure in the ML programming language developed by Damas and Milner \cite{DamasM82}. 

The existence of principal typings has also been investigated in various \emph{intersection type} disciplines.
Intersection types, pioneered by Coppo and Dezani 
\cite{CoppoD78}, increase the typability power of simple types
by assigning \emph{multiple} (traditionally, a \emph{set} of) types to terms. 
Intersection type systems can be tailored to characterize 
semantical classes of $\lambda$-terms such as the (head, strongly) normalizing ones, 
and to describe models of $\lambda$-calculus in a finitary way \cite{BarendregtCD83}. 
Due to the tight connections between intersection typability 
and term normalization, the typability problem in such systems is usually undecidable; 
nonetheless Coppo, Dezani and Venneri introduced an intersection type system 
which admits a notion similar to that of principal type \cite{CoppoDV80}. 
In this system, for each typable term $M$ it is possible to identify a \emph{principal pair} $(\Gamma, \iA)$
such that, for every derivable typing judgement $\Gamma' \der M : \iA'$,
the pair $(\Gamma', \iA')$ can be obtained from $(\Gamma, \iA)$ 
by means of substitution and a sequence of operations called \emph{expansions}. 

The need for this additional (and, in its original formulation, rather technical)
operation can be intuitively motivated  by observing that, 
in simple type assignment systems, all type derivations for a given term
share the same tree structure, and differ from each other only by the types occurring at their
nodes; on the contrary, in intersection type systems, derivations for the same term can differ both
in the previous sense \emph{and} in the structure of the derivation. 
Thus, substitution is used to match the types in the nodes of a derivation, 
while expansion adjusts its tree structure: 
starting from the derivation corresponding to $(\Gamma, \iA)$, 
which exhibits a minimal structure, 
expansions introduce additional subderivations, 
so to obtain a more complex structure corresponding to $(\Gamma',\iA')$. 
Summing up, expansion over pairs 
can be understood as the projection (on the conclusion) 
of structural changes applied to the whole derivation.

The already mentioned \cite{CoppoDV80} considers a system where 
intersection is seen as a non-idempotent connective, 
and types are \emph{strict} in the sense of \cite{Bakel92}, 
\ie intersection is not allowed on the right-hand side of arrows.
Ronchi Della Rocca and Venneri extended the notion of principal pair to a system
where intersection is idempotent, types are not strict, 
and they come with a preorder relation \cite{RoccaV84}. 
The system describes a filter model of $\lambda$-calculus,
and the preorder captures functional inclusion between elements of the model; 
in this setting, generating all pairs from the principal one 
requires three operations, namely substitution, 
expansion, and \emph{rise}, which deals with the preorder relation. 
Despite their differences, the systems of \cite{CoppoDV80} and \cite{RoccaV84}
have the same typability power, \ie they both characterize head normalization (see \cite{Bakel92}); thus, 
to keep the presentation as simple as possible, in this work we opt for a strict system.

Both \cite{CoppoDV80,RoccaV84} achieve their goals
exploiting the fact that the systems of interest enjoy an approximation theorem: a term can be
assigned all and only the types that can be assigned to its approximants, where approximants are
normal forms in a $\lambda$-calculus extended with a constant $\bot$ and associated reduction rules.
Ronchi Della Rocca later supplied the first type inference \emph{semi-algorithm} 
computing the principal pair of all and only the strongly normalizing terms \cite{Rocca88}.
Since then, multiple authors investigated the analogies 
between \emph{intersection type inference} and \emph{term reduction} 
\cite{NeergaardM04, CarlierW04, BoudolZ05, Boudol08}.

Intersection types have been primarily used for semantical purposes; 
still, in search of more practical applications, 
several restrictions for which the typability problem is decidable have been proposed. 
Many such restrictions are based on the notion of \emph{rank} \cite{Leivant83}, which intuitively measures 
the nesting of the intersection connective inside types. 
Remarkably, Kfoury and Wells showed that typability in all the finite
rank restrictions is decidable \cite{KfouryW99}. In order to prove this result in a purely syntactical
way, they  
reformulate the original notion of expansion and connect it to the structure of the
derivation, by enriching the syntax of types with the so-called \emph{expansion variables}, \ie pointers
to the subtypes that can be modified by expansions \cite{KfouryW04, CarlierW05, CarlierW12}.
In our work a similar role is played by a system of constraints that,
despite not being as flexible, avoids some of the bureaucracy needed to handle expansion variables.

Recent contributions that discuss principal typings in non-idempotent intersection
type systems include \cite{BernadetL13}, \cite{BucciarelliKV17} and \cite{Accattoli24}.
Compared to these works, which deal with the topic in a somewhat marginal way,
we delve deeper into algorithmic aspects of intersection type inference.

\paragraph{Contributions.} 
In this paper we supply a syntactic approach to principal typing in intersection type systems.
First, we introduce \emph{pseudo-derivations}, together
with three operations acting on them (substitution, expansion, erasure), 
and prove that such notions are \emph{correct and complete} w.r.t. typability 
in a system characterizing head normalization,
\ie they can be used to build any type derivation.
Second, we employ these ingredients in the design of an inference \emph{semi-algorithm}
computing the principal typing of the input term, in a subsystem characterizing strong normalization.
We place our work as an hopefully more accessible formulation 
of long-established results, with the aim of 
providing a gentle introduction to principal typings
and explaining the basic operations needed to construct them. 
To this end, we streamline some classical notions,
in particular that of expansion, putting additional emphasis 
on the parallelisms between the operations performed by our algorithm and term reduction.

\paragraph{Paper Organization.} The work at hand is structured as follows:
\begin{description}
\item[\Cref{sec:preliminaries}.] Preliminaries on $\lambda$-calculus 
and intersection type assignment systems $\systemIT$ and $\systemITsn$.
\item[\Cref{sec:intersection-types}.] 
Definition of pseudo-derivations, together with expansion and erasure operations. 
Correctness and completeness of such notions w.r.t. typability 
in $\systemIT$ and $\systemITsn$ (Theorems \ref{thm:correctness}, \ref{thm:completeness} and \ref{thm:corr-comp-strong}).
\item[\Cref{sec:algorithm}.] 
Design of a type inference semi-algorithm for $\systemITsn$.
Analysis of termination properties of the procedure (\Cref{thm:termination});
uniqueness and principality of the found solution (Theorems \ref{thm:algo-confluence} and \ref{thm:algo-principal}).
\item[\Cref{sec:conclusion}.] Conclusion and final remarks.
\end{description}

%% file: preliminaries.tex

\paragraph{\textbf{$\lambda$-calculus.}}
Terms and term contexts of \emph{$\l$-calculus} are generated by the following grammars:
\[  
M,N \ \grameq \ x \mid  \l x.M \mid MN \quad \quad \quad \quad 
\ccontext  \ \grameq \ \square \mid \l x. \ccontext  \mid  \ccontext M \mid M \ccontext 
 \] 
where $x$ ranges over a countable set of term variables. 
The abstraction $\l x.M$ binds $x$ in $M$; the writing $\FV(M)$ denotes the set of free variables of $M$.
We assume an
hygiene condition on variables: free and bound variables have
different names, and so do variables bound by different binders.
Given a context $\ccontext$ and a term $M$, the writing $\ccontext[M]$ denotes
the term obtained by replacing the single occurrence of the hole $\square$ in
$\ccontext$ by $M$, potentially capturing free variables of $M$.

The \emph{$\beta$-reduction}, denoted by $\redb$, is the contextual closure of the rewriting rule:  
\[
(\l x.M)N \mapsto_\beta M\subs{x}{N}
\] 
where $M\subs{x}{N}$ denotes the \emph{capture-free} substitution of $x$ by $N$ in $M$.
A term of shape $(\l x.M)N$ is called a $\beta$-redex. 
Such $\beta$-redexes are partitioned into $\ir$-redexes and $\kr$-redexes, 
depending on whether the variable $x$ occurs free in $M$ or not. 
Accordingly, we refer to the non-erasing part and the erasing part of $\beta$-reduction as
$\ir$-reduction and $\kr$-reduction. 

Given a binary reduction relation $\rightarrow_r$, we denote its 
transitive closure with $\rightarrow_r^+$, and its reflexive, 
transitive closure with $\rightarrow_r^*$.
A term is in \emph{$r$-normal form} when it contains no
$r$-redex; it is \emph{$r$-normalizing} if it can be reduced to a term in
$r$-normal form; it is \emph{strongly $r$-normalizing} if every $r$-reduction
sequence starting from it eventually stops. 

Terms in \emph{$\beta$-normal form} and \emph{head $\beta$-normal form}
are respectively generated by the following grammars (note that any normal form is an head normal form):
\[ 
K \ \grameq \ \l x.K  \gpipe  x K_1 \dots K_n 
\qquad \qquad 
H \ \grameq \ \l x.H  \gpipe  x M_1 \dots M_n  \qquad \qquad (n \geq 0)
\]
The notion of head normal form is the syntactical counterpart
of the notion of \emph{solvability} for (call-by-name) $\lambda$-calculus,
a term $M$ being solvable iff there is an \emph{head-context}
$\ccontext $ of shape $(\lambda x_1 \dots x_n. \square) M_1 \dots M_m$ 
such that $\ccontext[M]$ reduces to a fully determined result, 
traditionally the identity, \ie $\ccontext[M] \redb^* \lambda x.x$ \cite{Wadsworth76}. 
Solvable terms can be regarded as the ``meaningful'' terms that
produce some information, even if they may not be normalizing.

\paragraph{\textbf{Intersection Types.}}
Let us briefly recollect a \emph{non-idempotent} intersection type assignment system, 
based on the notion of \emph{multiset}. 
A multiset (notation $\ms, \mstwo$) is an unordered list of elements; 
we write $|\ms |$ for the cardinality of the multiset $\ms$, 
and $\uplus$ for the union of multisets taking into account multiplicities.

\begin{definition} \label{def:itypes}
The set $\IT$ of \emph{intersection types} is defined by the following grammar:
\[ 
\begin{array}{lrcl}
\textsc{Intersection Types} \ & \iA,\iB & \grameq & a \gpipe \ms \arrow \iA  \\
\textsc{Multisets} & \ms, \mstwo \ & \grameq & \mult{\iA_1,\dots,\iA_n} \quad (n \geq 0)
\end{array}
\]
where $a$ ranges over a countable set $\varIT$ of type variables.
\end{definition}

A \emph{type environment} is a total function from term variables to
multisets, such that only a finite number of variables is not mapped to the empty multiset. 
Environments are ranged over by $\Gamma$, $\Delta$.
The domain of an environment is $\dom{\Gamma} = \set{ x \mid \Gamma(x) \not = \emult}$; 
the union of environments is defined as
$(\Gamma \uplus \Delta)(x) = \Gamma(x) \uplus \Delta(x)$; 
it can be abbreviated $\Gamma, \Delta$ when $\Gamma$ and $\Delta$ have disjoint domains,
and $\Gamma, x : \ms$ is a special case of such notation.
$\Gamma \cminus x $ denotes the environment such that $(\Gamma \cminus x)(x) = \emult$ 
and $(\Gamma \cminus x)(y) = \Gamma(y)$ for every $y \not =  x$.
For convenience, the writing $\Gamma, x : \emult$ is considered to be the same as $\Gamma$.

\begin{definition} \label{def:int-system}
The \emph{intersection type assignment system} $\systemIT$,
assigning (multisets of) types to $\lambda$-terms, consists of the following rules:
\begin{gather*}
\infer[\axvar]{ x:\mult{\iA} \der x:\iA}{} 
\qquad \qquad
\infer[\abs]
{\Gamma \cminus x \der \lambda x. M: \Gamma(x) \arrow \iA}
{\Gamma \der M : \iA}
\\[5pt]
\infer[\many]
{\biguplus_{i=1}^n \Gamma_i \der M:\mult{\iA_1,\dots,\iA_n}}
{(\Gamma_i \der M: \iA_i)_{i=1}^n & n \geq 0 }
\qquad \qquad
\infer[\app]
{\Gamma \uplus \Delta \der MN: \iA}
{\Gamma \der M: \ms \arrow \iA & \Delta \der N : \ms}
\end{gather*}
\end{definition} 

Type derivations are ranged over by $\Pi, \Sigma$. 
We write $\Gamma \der M:\iA$ as a shorthand for the existence of a derivation proving 
$\Gamma \der M:\iA$, and to name a derivation with such conclusion we write $\Pi \dem \Gamma \der M:\iA$.

We also introduce the \emph{strong} system $\systemITsn$, a restriction of $\systemIT$ that forbids 
types containing empty multisets. To distinguish
between derivations in the two systems, 
typing judgments in $\systemITsn$ are noted $\Gamma \dersn M : \iA$.
The set $\ITsn \subset \IT$ is obtained  by imposing $n \geq 1$ in \Cref{def:itypes};
system $\systemITsn$ is obtained from \Cref{def:int-system} by imposing $n \geq 1$ 
in rule $\many$, and replacing rule $\abs$ by the two rules: 
\[
\infer[\absi]{\Gamma \dersn \lambda x.M : \ms \arrow \iA}
{\Gamma, x:\ms \dersn M : \iA  & \ms \not = \emult} 
\qquad \qquad 
\infer[\absk]{\Gamma \dersn \lambda x.M: \mult{\iA} \arrow \iB}
{\Gamma \dersn M : \iB & x \not \in \dom{\Gamma} & \iA \in \ITsn} 
\]

Systems $\systemIT$ and $\systemITsn$ are \emph{relevant} (\ie no unnecessary weakening is allowed), 
and coincide with systems $\mathcal{W}$ and $\mathcal{S}$ of \cite{BucciarelliKV17};
the following folklore properties are inherited from there. 

\begin{theorem}[Subject Reduction/Expansion] \label{thm:subconv} \
Let $M \redb N$. Then $\Gamma \der M : \iA$ if and only if $\Gamma \der N : \iA$. 
\end{theorem}

\begin{theorem}[Characterization] \label{thm:characterization} \
\begin{itemize}
\item $M$ is $\systemIT$-typable if and only if it is head $\beta$-normalizing.
\item $M$ is $\systemITsn$-typable if and only if it is strongly $\beta$-normalizing.
\end{itemize}
\end{theorem}

\begin{remark} \label{rmk:strong-systems}
Neither subject reduction nor subject expansion hold in the strong system $\systemITsn$. 
Absence of subject reduction is due to the relevance of the system,
and can be recovered by adding a weakening rule. Note that a weaker property holds:
if $M$ is $\systemITsn$-typable and $M \redb N$, 
then $N$ is $\systemITsn$-typable, possibly with a different typing. 
For example, consider the derivation $x : \mult{a,b} \dersn (\l y.x)x : a$. 
Although $x : \mult{a,b} \dersn x : a$ is not derivable, one can derive $x : \mult{a} \dersn x : a$. 
Concerning subject expansion, $\systemITsn$-typability is preserved (possibly with a different typing) 
only if one considers expansion w.r.t. \emph{specific} notions of reduction, namely those
that preserve strong normalization; unrestricted $\beta$-reduction does not fall in this category,
because of the erasing $\kr$-reduction steps. For example, consider
$(\l y.x)(\delta \delta) \redb x$, where $\delta = (\l z.zz)$. Clearly $x$ is $\systemITsn$-typable, but  
$(\l y.x)(\delta \delta)$ is not, since $\delta \delta$ can only be typed using the empty multiset.  
\end{remark}

%% file: intersection.tex

To better reason about 
(non-idempotent) intersection type systems, and significantly reduce
the complexity of the procedures we discuss in the next sections, 
it is convenient to impose some kind of ordering to multisets.
Formally, we need the notion of \emph{intersection pre-types}.

Intersection pre-types (pre-types for brevity) are analogous to intersection
types, the only difference being that multisets are replaced by lists.
Lists are ranged over by $\ls, \lstwo$; we write $|\ls|$
for the length of the list $\ls$, and the symbol $\cons$ denotes list concatenation.
The set $\PT$ of pre-types is defined by:
\[ 
\begin{array}{lrcl}
\textsc{Intersection Pre-types} \ & \A,\B & \grameq & \con \; \mid \; \ls \arrow \A  \\
\textsc{Lists} & \ls, \lstwo & \grameq & \llist{\A_1, \dots, \A_n} \quad (n \geq 0) \\ 
\end{array}
\]
where $\con$ ranges over a countable set $\varPT$ of pre-type variables.  
Remark that the distinction between pre-type variables (notation $\con$) and type variables
(notation $a$) is not a meaningful one; the two different notations are used only to help visually
distinguish pre-types from actual intersection types.

A \emph{pre-type environment} is a total function from term variables to lists,
such that only a finite number of variables is not mapped to the empty list.
Pre-type environments are ranged over by $\tGamma, \tDelta$,
and $\dom{\tGamma} = \set{x \mid \tGamma(x) \not = \ellist}$.
Union of pre-type environments is
defined as $(\tGamma \cdot \tDelta)(x) = \tGamma(x) \cdot \tDelta(x)$.
We extend all notations previously introduced for type environments
to pre-type environments.

\begin{figure}
\begin{gather*}
\infer[\axvar]{\Pi \dem x:\llist{\con} \derp x:\con}{} \qquad \eqset_\Pi = \emptyset
\\[10pt]
\infer[\abs]
{\Pi \dem \tGamma \cminus x \derp \lambda x. M: \conb}
{\Sigma \dem \tGamma \derp M:\con & \conb \text{ fresh}}
\qquad
\eqset_\Pi = \eqset_\Sigma \cup \{ \conb \doteq \tGamma(x) \arrow \con \}
\\[10pt]
\infer[\many]
{\Pi \dem \tGamma_1 \cons \dots \cons \tGamma_n \derp M : \llist{\con_1,\dots,\con_n}}
{(\Sigma_i \dem \tGamma_i \derp M : \con_i)_{i=1}^n &
\forall i \not = j.\Sigma_i * \Sigma_j  & n \geq 0 }
\qquad
\begin{array}{l}
\eqset_\Pi =\bigcup_{i=1}^n \eqset_{\Sigma_i}
\end{array}
\\[10pt]
\infer[\app]
{\Pi \dem \tGamma \cons \tDelta \derp MN: \conc}
{\Sigma_1 \dem \tGamma \derp M: \con & \Sigma_2 \dem \tDelta \derp N: \llist{\conb_1,\dots,\conb_n} & \Sigma_1 * \Sigma_2 & \conc \text{ fresh}}
\\
\eqset_\Pi = \eqset_{\Sigma_1} \cup \eqset_{\Sigma_2} \cup \{\con \doteq \llist{\conb_1,\dots,\conb_n} \arrow \conc \}
\end{gather*}
\caption{Pseudo-derivation rules for system $\systemIT$.} \label{fig:pseudo}
\end{figure}

\subsection{Pseudo-derivations}

A \emph{pseudo-derivation} is a tree of judgements assigning (lists of) pre-types to terms; 
to each pseudo-derivation $\Pi$ is associated a set of 
equations between (lists of) pre-types $\eqset_{\Pi}$, 
which keeps track of all and only 
the constraints that must be satisfied in order 
to transform $\Pi$ into an actual derivation.
Since $\eqset_\Pi$ can be recovered from $\Pi$, we 
sometimes leave $\eqset_\Pi$ implicit. Moreover, when $\Pi$ is clear from the context, 
we may write $\eqset$ instead of $\eqset_\Pi$ for the system of equations associated to $\Pi$.
Remark that pseudo-derivations are defined modulo renaming of pre-type variables,
and that premises of each rule are \emph{ordered}.

Letting $\Tvar(\A)$ denote the set of pre-type variables occurring in $\A$, we say that $\A$ and $\B$ are
\emph{disjoint} (written $\A * \B$) if $\Tvar(\A) \cap \Tvar(\B) = \emptyset$. 
The notion of disjointness is extended to environments and derivations
in the standard way. A type $\A$ is \emph{fresh} w.r.t. a 
derivation $\Pi$ if $\A * \B$ for each $\B$ occurring in $\Pi$.

\begin{definition}[Pseudo-derivations for $\systemIT$] \label{def:pseudo} \
\begin{itemize}
\item A \emph{pseudo-derivation} for $M$ is a pair $\PD{M} = (\Pi, \eqset_\Pi)$,
where $\Pi$ is a tree of judgements assigning (lists of) pre-type variables to terms and $\eqset_{\Pi}$ 
is the associated system of equations, as per \Cref{fig:pseudo}.
\item The \emph{minimal pseudo-derivation} for $M$ is a pseudo-derivation for $M$
such that $n = 1$ in all rules $\many$.
Since the minimal pseudo-derivation is unique, 
modulo renaming of pre-type variables, we refer to it as $\PDmin{M}$.
\end{itemize}
\end{definition}

Next we introduce notions related to the search 
for solutions of a set of equations between (lists of) pre-types. 
The symbol $\phi$ stands for a substitution $\varIT \arrow \IT$, 
while $\psi$ stands for a substitution $\varPT \arrow \PT$; 
we pose $\dom{\psi} = \set{\con \in \varPT \mid \psi(\con) \not = \con}$ and
$\cod{\psi} = \set{\A \mid \con \in \dom{\psi} \text{ and } \psi(\con) = \A }$. 
Substitutions are extended 
to  types (resp. pre-types), 
multisets (resp. lists) and derivations in the standard way.

\begin{definition} \label{def:solvability} 
Let $S=\set{\A_i \doteq \B_i \mid i \in I} \cup \set{\ls_j \doteq \lstwo_j \mid j \in J}$
be a set of equations between (lists of) pre-types.
\begin{itemize}
\item $\psi : \varPT \arrow \PT$ \emph{solves} $S$
if $\psi (\A_i) = \psi (\B_i)$ for all $i \in I$
and $\psi (\ls_j) = \psi (\lstwo_j)$ for all $j \in J$. 
\item $S$ is in \emph{solved form} if the following conditions are met:
\begin{itemize}
	\item $J = \emptyset$, \ie there is no equation between lists; 
	\item every $\A_i$ is a variable $\con_i$, and all variables $\con_i$ are distinct;
	\item no left-hand side $\con_i$ appears in some right-hand side $\B_k$.
\end{itemize}
\item $S$ is in \emph{unsolvable form} if it contains at least one \emph{circular equation},
\ie an equation of shape $\con \doteq \A$ such that $\con$ occurs in $\A \not = \con$.
\item $S$ is in \emph{blocked form} if it contains at least one \emph{blocked equation},
\ie an equation between lists $\ls \doteq \lstwo$ such that $|\ls| \not= |\lstwo|$.
\end{itemize}
\end{definition}

\begin{notation}
If $S = \set{ \con_i \doteq \B_i \mid i \in I}$ is in solved form, 
$\mgu{S}$ denotes the most general substitution solving $S$ 
(also called the most general unifier of $S$), that is $\mgu{S}(\con_i) = \B_i$. 
\end{notation}

\Cref{fig:unification-intersection} introduces the unification
rules used to try to solve a system of equations between (lists of) pre-types.
The writing $S\subs{\con}{\A}$ denotes the set obtained 
from $S$ replacing \emph{every} occurrence of $\con$ by $\A$.

\begin{figure}[b]
\begin{gather*}
\infer[\uerase]
{ S }
{ S \cup \{\A \doteq \A\} 
}
\qquad
\infer[\uswap]
{ S  \cup \{\con \doteq \ls \arrow \A\} }
{ S \cup \{\ls \arrow \A \doteq \con \} }
\qquad 
\infer[\uarrow]
{ S  \cup \{\ls \doteq \lstwo\} \cup \{\A \doteq \B\} }
{ S \cup \{\ls \arrow \A \doteq \lstwo \arrow \B\} }
\\[8pt]
\infer[\ulist]
{ S  \cup \{\A_i \doteq \B_i\}_{1 \leq i \leq n} }
{ S \cup \{\llist{\A_1, \ldots, \A_n}  \doteq \llist{\B_1, \ldots, \B_n} \} & n \geq 0 }
\qquad 
\infer[\usubs]
{ S[\A / \con] \cup \{\con \doteq \A\} }
{ S \cup \{\con \doteq \A\} & \con \notin \Tvar(\A) & \con \in \Tvar(S) }
\end{gather*}
\caption{Unification rules for intersection pre-types.}
\label{fig:unification-intersection}
\end{figure}

\begin{definition} \
\begin{itemize} 
\item $S \rewu S'$ means that $S'$
is obtained from $S$ by applying one of the rules of \Cref{fig:unification-intersection}. 
\item A set of equations $S$ is in $\rewu$-\emph{normal form} if no rule can be applied to it.
\end{itemize}
\end{definition}

The following Properties \ref{prop:unif-confluence} and \ref{prop:unif-nf} 
come from the fact that the rules of \Cref{fig:unification-intersection}
are Robinson's unification rules \cite{Robinson65,MartelliM82}, instantiated 
to the grammar of intersection pre-types.
In particular, writing $\nf{S}$ for the unique $\rewu$-normal form of $S$,
we remark that $\nf{S}$ is either in solved form, or in unsolvable/blocked form
(possibly both unsolvable \emph{and} blocked at the same time).

\begin{property} \label{prop:unif-confluence}
$\rewu$ is terminating and confluent (modulo renaming of pre-type variables).
\end{property}

\begin{property} \label{prop:unif-nf}
$\nf{S}$ is either in solved form, or in unsolvable/blocked form.
\end{property}

\subsection{Expansion and Erasure}\label{sub:expansion}
Applying unification rules to the set of equations associated 
to a pseudo-derivation may result in a blocked form.
In order to deal with such a scenario, we introduce two operations 
called \emph{expansion}\footnote{Not to be confused with \emph{subject expansion}. 
The name has been chosen for historical reasons discussed in the introduction.} 
and \emph{erasure}. Both modify the tree structure of the pseudo-derivation 
they are applied to: an expansion increases the number of premises of a rule $\many$, 
whereas an erasure decreases it.

\begin{definition}[Expansion, Erasure] \label{def:expansion} 
Let $(\Pi, \eqset)$ be a pseudo-derivation. 
\begin{itemize}
\item 
	An \emph{expansion} operation, written $\Expand{\ls,n}$, has two parameters: 
	a list $\ls \not = \ellist$ and a natural number $n \geq 1$.
	The result of applying $\Expand{\ls,n}$ to $(\Pi, \eqset)$, 
	for brevity written $\Expand{\ls,n,\Pi}$, is a pseudo-derivation $(\Pi', \eqset')$ 
	such that:
	\begin{itemize}
	\item If $\ls = \llist{\con_1,\dots,\con_m}$ and $\Pi$ contains a rule
	\[
	\infer[\many]
	{\tDelta_1 \cons \dots \cons \tDelta_m \derp N: \llist{\con_1,\dots,\con_m} }
	{(\tDelta_i \derp N:\con_i)_{i=1}^m & }
	\]
	then $\Pi'$ has the same tree structure as $\Pi$,
	with the exception of the subtree whose root is the above $\many$ rule, 
	which is replaced by:
	\[
	\infer[\many]
	{\tGamma_1 \cons \dots \cons \tGamma_{m+n} \derp N: \llist{\con_1,\dots,\con_m,\dots,\con_{m+n}} }
	{(\Sigma_i \dem \tGamma_i \derp N : \con_i)_{i=1}^{m+n}}
	\]
	where $\Sigma_{1},\dots,\Sigma_{m+n}$ are fresh disjoint copies of $\PDmin{N}$.
	\item Otherwise, if no such rule exists, $\Pi' = \Pi$.
	\end{itemize}
\item 
	An \emph{erasure} operation, written $\Erase{\ls,n}$, 
	has two parameters: a list $\ls \not = \ellist$ and a natural number $n \geq 1$. 
	The result of applying $\Erase{\ls,n}$ to $(\Pi, \eqset)$, 
	for brevity written $\Erase{\ls,n,\Pi}$, is a pseudo-derivation $(\Pi', \eqset')$ 
	such that:
	\begin{itemize}
	\item If $\ls = \llist{\con_1,\dots,\con_m}$ and (as in the previous point) $\Pi$ contains a $\many$ rule with
	conclusion $\tDelta_1 \cons \dots \cons \tDelta_m \derp N: \llist{\con_1,\dots,\con_m}$,
	then $\Pi'$ has the same tree structure as $\Pi$,
	with the exception of the subtree whose root is the above $\many$ rule, 
	which is replaced by:
	\[
	\infer[\many]
	{\tGamma_1 \cons \dots \cons \tGamma_{m-n} \derp N: \llist{\con_1,\dots,\con_{m-n}} }
	{(\Sigma_i \dem \tGamma_i \derp N : \con_i)_{i=1}^{m-n}}
	\]
	where $\Sigma_{1},\dots,\Sigma_{m-n}$ are fresh disjoint copies of $\PDmin{N}$. 
	If $m \leq n$, the resulting $\many$ rule has conclusion $\derp N : \ellist$.
	\item Otherwise, if no such rule exists, $\Pi' = \Pi$.
	\end{itemize}
\end{itemize}
\end{definition}

It is easy to check that the above \Cref{def:expansion} is well posed, that is, 
for every pseudo-derivation $\PD{M}=(\Pi, \eqset)$, natural $n \geq 1$, and list $\ls \not = \ellist$, 
both $\Expand{\ls,n,\Pi}$ and $\Erase{\ls,n,\Pi}$ are pseudo-derivations for $M$.
Thanks to expansion and erasure operations, 
the notion of solvability can be extended to pseudo-derivations.

\begin{definition} \label{def:pd-solution}
Let $(\Pi, \eqset)$ be a pseudo-derivation. A solution of $(\Pi, \eqset)$ is a pair $(\expseq,\psi)$, 
where $\expseq$ is a sequence of expansions and erasures such that 
$\expseq(\Pi,\eqset) = (\Pi', \eqset')$
and $\psi : \varPT \arrow \PT$ is a solution of $\eqset'$. 
\end{definition}

Note that arbitrary expansions and erasures do \emph{not} preserve 
the solvability of a pseudo-derivation, as illustrated by \Cref{exmp:expand-erase-not-safe}.
This suggests a design principle we will follow in our type inference algorithm
(see later \Cref{sec:algorithm}):
modifications to the structure of a pseudo-derivation should always 
be guided by the system of equations associated to it.

\begin{example} \label{exmp:expand-erase-not-safe}
Let $M=(\lambda x.x)y$. Then $\PDmin{M}=(\Pi, \eqset)$, where $\Pi$ is:
\[
\infer{ \Pi \dem y:\llist{\conc}\derp (\lambda x.x)y:\cond }
{\infer{\derp \lambda x.x: \conb}{x:\llist{\con} \derp x:\con} & \infer{y:\llist{\conc}\derp y:\llist{\conc} }{y:\llist{\conc}\derp y:\conc }
}
\]
and $\eqset=\{\conb\doteq \llist{\con} \arrow \con, \conb=\llist{\conc}\arrow \cond \}$. 
The set $\eqset$ reduces to
$\{\conb \doteq \llist{\con}\arrow \con, \conc \doteq \con, \cond \doteq \con \}$,
which is in solved form.
Now consider the two following scenarios:
\begin{itemize}
\item Let $\Expand{\llist{\conc}, 1, \Pi} = (\Pi',\eqset')$, where $\Pi'$ is:
\[
\infer{\Pi' \dem y:\llist{\conc, \conc'} \derp (\lambda x.x)y:\cond }
{\infer{\derp \lambda x.x: \conb}{x:\llist{\con} \derp x:\con} & \infer{y:\llist{\conc, \conc'}\derp y:\llist{\conc, \conc'} }{y:\llist{\conc}\derp y:\conc & y:\llist{\conc'}\derp y:\conc'}
}
\]
and $\eqset' = \{\conb\doteq \llist{\con} \arrow \con, \conb=\llist{\conc,\conc'}\arrow \cond \}$. 
$\eqset'$ reduces to 
$\{\conb \doteq \llist{\con}\arrow \con, \cond\doteq\con, \llist{\con} \doteq \llist{\conc, \conc'} \}$, 
which is in blocked form. 
But expansion cannot get rid of the block, since $ \llist{\con}$ does not occur in the conclusion of a rule $\many$, 
and consequently $\Expand{ \llist{\con},1,\Pi'} = (\Pi', \eqset')$.
The situation can only be unblocked performing $\Erase{\llist{\conc,\conc'},1,\Pi'}$.
\item Let $\Erase{\llist{\conc},1,\Pi} = (\Pi',\eqset')$, where $\Pi'$ is:
\[
\infer{ \Pi' \dem \derp (\lambda x.x)y:\cond }
{\infer{\derp \lambda x.x: \conb}{x:\llist{\con} \derp x:\con} 
 & \derp y: \ellist}
\]
and $\eqset' = \{\conb\doteq \llist{\con} \arrow \con, \conb=\ellist \arrow \cond \}$.
$\eqset'$ reduces to $\{\conb\doteq \llist{\con} \arrow \con, \cond \doteq \con,  \llist{\con}\doteq\ellist \}$,
which cannot be unblocked neither by expansions nor erasures, as $\Erase{\llist{\con},1,\Pi'} = (\Pi', \eqset')$. 
\end{itemize} 
\end{example}

\subsection{Correctness and Completeness}
In this section we prove that the notion of (minimal) pseudo-derivation, together with 
expansion, erasure, and substitution operations, is \emph{correct and complete} w.r.t. typability in system $\systemIT$.
This guarantees that such basic ingredients are suitable
for the search for (principal) typings in the systems of interest.

In order to transform a pre-type into an actual intersection type,
we define the following two mutually recursive functions 
whose domains are, respectively, pre-types and lists 
(for simplicity, we use the symbol $\ltom$ for both of them):
\[
\ltom(\con)  =  a \qquad \quad 
\ltom(\ls \arrow \A) =  \ltom(\ls) \arrow \ltom(\A) \qquad \quad 
\ltom(\llist{\A_1, \dots, \A_n})  =  \mult{\ltom(\A_1),\dots,\ltom(\A_n)} 
\]
where we assume that $\ltom$ is injective on pre-type variables. 
The function $\ltom$ is extended to pre-type environments and derivations in the standard way.

Using $\ltom$, first we show that, 
starting from any solution of a
pseudo-derivation, one obtains an infinite family of derivations in system $\systemIT$.

\begin{theorem}[Correctness] \label{thm:correctness}
Let $\PD{M}=(\Pi,\eqset)$ be a pseudo-derivation, $(\expseq,\psi)$ be a solution of $(\Pi,\eqset)$,
and $\expseq(\Pi) \dem \tGamma \derp M : \con$.
Then $\phi \circ \ltom \circ \psi \circ \expseq(\Pi) \dem \phi \circ \ltom \circ \psi (\tGamma) \der M : \phi \circ \ltom \circ \psi (\con) $ for all $\phi : \varIT \arrow \IT$.
\end{theorem}

\begin{proof}
The fact that $\ltom \circ \psi \circ \expseq(\Pi) \dem \ltom \circ \psi (\tGamma) \der M : \ltom  \circ \psi (\con) $
comes from \Cref{def:pd-solution}. The result follows, 
as derivations are closed under substitution of type variables.
\end{proof}

Second, we show that all derivations for a term $M$ in system $\systemIT$ 
can be obtained starting from the \emph{minimal} pseudo-derivation $\PDmin{M}$. 
The proof of this property is split in two steps: 
the first point of \Cref{lem:from-pd-to-der}
shows that any pseudo-derivation can be obtained from the minimal
one via a suitable sequence of erasures and expansions, while the second point states that 
any derivation $\Sigma$ can be obtained by substitution from a pseudo-derivation $\Pi$
sharing the same structure as $\Sigma$.

\begin{lemma} \label{lem:from-pd-to-der} \
\begin{enumerate}
\item Let $\PD{M} = (\Pi, \eqset)$ and $\PDmin{M}=(\Pi_M, \eqset_M)$. Then there is 
a sequence of expansions and erasures $\expseq$ such that 
$(\Pi, \eqset)=\expseq(\Pi_M, \eqset_M)$.
\item Let $\Sigma \dem \Gamma \der M:\iA$. Then there are a pseudo-derivation $\PD{M}=(\Pi, \eqset)$ 
and a substitution $\psi$ such that $\psi$ solves $\eqset$ and $\Sigma = \ltom \circ \psi(\Pi)$.
\end{enumerate}
\end{lemma} 

\begin{proof}
See \Cref{sec:completeness-appendix}.
\end{proof}

\begin{theorem}[Completeness] \label{thm:completeness}
Let $\Sigma \dem \Gamma \der M : \iA$ and $\PDmin{M} = (\Pi_M,\eqset_M)$.
Then there is a solution $(\expseq,\psi)$ of $(\Pi_M, \eqset_M)$ such that 
$\Sigma = \ltom \circ \psi \circ \expseq(\Pi_M)$.
\end{theorem}

\begin{proof}
Immediate consequence of \Cref{lem:from-pd-to-der}.
\end{proof}

\subsection{The Strong Case}

This brief section recasts some notions to fit the setting of system $\systemITsn$.
We start from pseudo-derivations, which in the strong case do not allow empty lists in the subject type. 

\begin{definition}[Pseudo-derivations for $\systemITsn$] \label{def:strong-pseudo} \
\begin{itemize}
\item A \emph{strong pseudo-derivation} for $M$ is built as per the rules of \Cref{fig:pseudo}, but
imposing $n \geq 1$ in rule $\many$ and replacing rule $\abs$ by the two rules:
\begin{gather*}
\vcenter{ 
\infer[\absi]{\Pi \dem \tGamma \derp \lambda x. M : \conb}
{\Sigma \dem \tGamma, x:\ls \derp M : \con  & \ls \not = \ellist &\conb \text{ fresh}} 
} \qquad 
\eqset_\Pi = \eqset_{\Sigma} \cup \set{\conb \doteq \ls \arrow \con} 
\\[3pt]
\vcenter{ 
\infer[\absk]{\Pi \dem \tGamma \derp \lambda x. M: \conc}
{\Sigma \dem \tGamma \derp M:\con & x \not \in \dom{\tGamma} & \conb,\conc \text{ fresh}} 
} \qquad 
\eqset_\Pi = \eqset_\Sigma \cup \set{\conc \doteq \llist{\conb} \arrow \con} 
\end{gather*}
\item A \emph{minimal strong pseudo-derivation} for $M$ is a strong pseudo-derivation such that $n=1$
in all rules $\many$. Since the minimal strong pseudo-derivation is unique, 
modulo renaming of pre-type variables, we refer to it as $\PDsnmin{M}$.
\end{itemize}
\end{definition}

To adapt the expansion operation to strong pseudo-derivations,
it suffices to replace $\PDmin{N}$ by $\PDsnmin{N}$ in the first point of \Cref{def:expansion}.
On the other hand, erasure is not needed.
This means that a solution of $\PDsnmin{M}$ consists of a sequence of \emph{expansions only},
together with a final unifying substitution. 
\Cref{thm:correctness}, \Cref{lem:from-pd-to-der} and \Cref{thm:completeness} 
are also easily adapted to the {strong setting}:

\begin{lemma} \label{lem:from-pd-to-der-strong} \
\begin{enumerate}
\item Let $\PDsn{M} = (\Pi, \eqset)$ and $\PDsnmin{M}=(\Pi_M, \eqset_M)$. Then there is 
a sequence of expansions $\expseq$ such that $(\Pi, \eqset) = \expseq(\Pi_M, \eqset_M)$.
\item Let $\Sigma \dem \Gamma \dersn M: \iA$. Then there are a strong pseudo-derivation $\PDsn{M}=(\Pi, \eqset)$ 
and a substitution $\psi$ such that $\psi$ solves $\eqset$ and $\Sigma = \ltom \circ \psi(\Pi)$.
\end{enumerate}
\end{lemma} 

\begin{theorem}[Correctness and Completeness in the Strong Case] \label{thm:corr-comp-strong} \
\begin{itemize}
\item Let $\PDsn{M}=(\Pi,\eqset)$ be a strong pseudo-derivation, 
$(\expseq,\psi)$ be a solution of $(\Pi,\eqset)$,
and $\expseq(\Pi) \dem {\tGamma \derp M : \con}$.
Then $\phi \circ \ltom \circ \psi \circ \expseq(\Pi) \dem \phi \circ \ltom \circ \psi (\tGamma) \dersn M : \phi \circ \ltom \circ \psi (\con) $ for all $\phi : \varIT \arrow \IT$.
\item Let $\Sigma \dem \Gamma \dersn M : \iA$ and $\PDsnmin{M} = (\Pi_M,\eqset_M)$.
Then there is a solution of $(\expseq,\psi)$ of $(\Pi_M, \eqset_M)$ such that 
$\Sigma = \ltom \circ \psi \circ \expseq(\Pi_M)$.
\end{itemize}
\end{theorem}

%% file: algo.tex

Building upon the notion of expansion, 
in this section we design a \emph{semi-algorithm} that, taken in input a term $M$,
tries to solve the system of constraints generated by $\PDsnmin{M}$. 
We then prove that the algorithm yields a solution if and only if $M$ is strongly $\beta$-normalizing. 

To uniquely identify the $\many$ rule that needs to be modified by an expansion (resp. erasure), 
\Cref{def:expansion} relies on on the fact that, by construction, each subtree of 
a pseudo-derivation contains disjoint pre-type variables. In order to preserve this disjointness property
throughout algorithm execution, we introduce a slight modification of the standard unification rules: 
specifically, we do not replace occurrences of pre-type
variables that are contained inside lists, so that no ambiguity can arise.

\begin{definition}
$S \rewo S'$ means that $S'$ is obtained from $S$ 
by applying one of the rules in \Cref{fig:unification-intersection}, 
but replacing rule $\usubs$ by rule:
\[
\infer[\usubszero]
{ S\zerosubs{\con}{\A} \cup \{\con \doteq \A\} }
{ S \cup \{\con \doteq \A\} & \con \notin \Tvar(\A) & \con \in \outvar(S) }
\]
where  $\outvar(S) \subseteq \Tvar(S)$
is the set of all pre-type variables that occur \emph{outside} a list in $S$, and 
$S\zerosubs{\con}{\A}$ denotes the set obtained 
from $S$ replacing \emph{only} the occurrences of $\con$
that are not contained inside a list.
\end{definition}

Notice that $S$ may be in $\rewo$-normal form but
neither in solved, nor in unsolvable, nor in blocked form,
as testified by $S = \{ \con \doteq \llist{\conb} \arrow \conc, \conb \doteq \llist{\con} \arrow \cond \}$.  
Still, is it possible to isolate some good properties that 
relate the behaviour of $\rewo$ to that of $\rewu$. 

\begin{property} \label{prop:unif-out-confluence}
$\rewo$ is terminating and confluent (modulo renaming of pre-type variables). 
\end{property}

\begin{property} \label{prop:unif-out} \
\begin{itemize}
\item Writing $\nfo{S}$ for the unique $\rewo$-normal form of $S$, one has $\nfo{S} \rewur \nf{S}$.
\item If $\nfo{S}$ is not in blocked form, then $\nf{S}$ is not in blocked form
(that is, $\nf{S}$ is either in solved or unsolvable form).
\end{itemize}
\end{property}

\begin{proof}
First observe that if $\nfo{S}$ is not in blocked form, then it must have shape:
\[ \nfo{S} = 
\set{ \con_i \doteq \ls^i_1 \arrow \dots \arrow  \ls^i_{m_i} \arrow \conb_i }_{1 \leq i \leq n}
\quad (m_i \geq 0) 
\] 
Moreover $\con_1,\dots,\con_n$ are all distinct, 
and if $\set{\con_1,\dots,\con_n} \cap \set{\conb_1,\dots,\conb_n} \not = \emptyset$ 
then $\nfo{S}$ is in unsolvable form. Indeed, we cannot have $\con_i = \con_j$ for $i \not = j$, 
otherwise $\usubszero$ could be applied. 
In addition, we cannot have $\con_i = \conb_j$ for $i \not = j$, 
otherwise $\usubszero$ could substitute the tail occurrences of $\conb_j$ in the RHS. 
The only scenario where $\set{\con_1,\dots,\con_n} \cap \set{\conb_1,\dots,\conb_n} \not = \emptyset$
but neither $\uerase$ nor $\usubszero$ can be applied is 
if $\con_i = \conb_i$ for at least one $i$ such that $m_i \geq 1$, in which case $\nfo{S}$ is in unsolvable form.
Lastly observe that, in order to go from $\nfo{S}$ to $\nf{S}$, it suffices to apply $\usubs$ rules
replacing the occurrences of LHS variables that are found inside lists, 
which $\usubszero$ could not replace. As these substitutions cannot generate blocked equations, we conclude.
\end{proof}



\subsection{The Semi-algorithm}
We are now ready to discuss the semi-algorithm $\InferStrong$. 
The input of $\InferStrong$ is a term $M$; the output, if it terminates, is a pair $(\Pi, \psi)$, where
$\Pi$ is a pseudo-derivation for $M$ whose associated set of equations $\eqset$ is such that 
$\nf{\eqset}$ is in solved form, and $\psi$ is the substitution induced by it.

\begin{notation} 
$\Blocked(S) = \mathtt{true}$ if and only if the set of equations $S$ is in blocked form.
\end{notation}

\begin{algorithm}
\caption{The non-deterministic semi-algorithm $\InferStrong$. Input: a term $M$.}
\label{algorithm}
\begin{algorithmic}[1]
\Function{$\InferStrong(M)$}{}
	\State $(\Pi, \eqset) \gets \PDsnmin{M}$
	\State $\eqset \gets \nfo{\eqset}$ 
	\While {$\Blocked(\eqset)$}  \label{alg:while-begin}
		\State \textbf{choose} $ (\ls \doteq \lstwo) \in \eqset$ such that $|\ls | \not= |\lstwo |$
			\If {$ |\ls | > |\lstwo | $}
				\State $n \gets |\ls |-|\lstwo |$
				\State $(\Pi,\eqset) \gets \Expand{\lstwo, n, \Pi}$ 
			\Else 
				\State $n \gets |\lstwo|-|\ls | $
				\State $(\Pi,\eqset) \gets \Expand{\ls, n, \Pi}$ 
			\EndIf
		\State $\eqset \gets \nfo{\eqset}$ 
	\EndWhile \label{alg:while-end}
	\State $\eqset \gets \nf{\eqset}$
	\State $\psi \gets \mgu{\eqset}$
	\State \Return $(\Pi,\psi)$ 
\EndFunction
\end{algorithmic} 
\end{algorithm} 


The algorithm operates on pairs of shape $(\Pi,\eqset)$, consisting of a pseudo-derivation 
and a (possibly partially reduced) set of equations.
Extending $\rewo$ to such pairs, we write $(\Pi,\eqset) \rewo (\Pi,\eqset')$ if 
$\eqset \rewo \eqset'$. 
The writing $(\Pi,\eqset) \rewe (\Pi',\eqset')$ means that $(\Pi',\eqset')$ is
the result of an expansion that $\InferStrong$
can perform on $\Pi$ to try to unblock a blocked equation in $\eqset$,
as per the algorithm description. 
Lastly, $(\Pi,\eqset) \Rew (\Pi',\eqset')$ means that $(\Pi,\eqset) \rewor (\Pi,\nfo{\eqset}) \rewe (\Pi',\eqset')$.
Adopting these conventions, a run of $\InferStrong$ 
can be understood as a sequence of $\Rew$ steps, followed by a sequence of $\rewu$ steps.

Note that, since the system of equations is recomputed
at each expansion\footnote{Of course, in an actual implementation, 
efficiency could be improved by keeping track of the modified constraints only, 
so to avoid repeating unification steps. 
Here we favoured clarity of exposition over efficiency.}, 
the only purpose of the intermediate $\rewo$-reduction steps is to expose blocked
equations, thus guiding the structural changes to the pseudo-derivation.
The algorithm is non-deterministic, as blocked equations are randomly chosen.

We point out that the algorithm checks whether the set of equations 
is in blocked form, but does not consider the possibility 
it may be in unsolvable form. The following \Cref{lem:eqset-not-unsolvable}, 
of which we do not provide a formal proof, 
motivates this design choice by stating that applying unification rules to the system 
of constraints associated to a pseudo-derivation cannot result in an unsolvable form.
Formally proving this statement is surprisingly difficult, 
and interestingly none of the works dealing with intersection typability as
an extended unification problem seems to acknowledge the 
relevance of such a result, even when they implicitly rely on it.
An intuitive justification can be provided by comparing \emph{simple} and \emph{intersection} type systems.
Indeed, when checking for simple typability of a given term, circular equations may be generated;
this is essentially due to the strict conditions imposed 
on type environments, which are required to agree on common variables.
In contrast, intersection type systems  
impose no constraint on the union of environments, thus preventing circularities.  

\begin{lemma} \label{lem:eqset-not-unsolvable}
If $\PDsn{M} = (\Pi,\eqset)$, then $\nf{\eqset}$ is not in unsolvable form.
\end{lemma}

\begin{theorem}[Algorithm Correctness] \label{thm:algo-correctness}
Let $\InferStrong(M)=(\Pi, \psi)$, where $\Pi \dem \tGamma \derp M: \con$. Then
$\ltom \circ \psi(\Pi) \dem \ltom \circ \psi (\tGamma) \dersn M : \ltom  \circ \psi (\con) $.
\end{theorem}

\begin{proof}
By the first point of \Cref{thm:corr-comp-strong}, 
the second point of \Cref{prop:unif-out} and \Cref{lem:eqset-not-unsolvable}.
\end{proof}

\begin{remark}
$\InferStrong$ does not find all possible solutions of $\PDsnmin{M}$, 
because expansions are used only when strictly necessary (\ie to unblock a blocked equation),
and in a minimal way (\ie without enlarging lists more than required).
\end{remark}

\subsection{Termination} 

Since $\InferStrong$ is non-deterministic, we say that 
the algorithm terminates if \emph{there exists} a terminating execution path.
We now show that if $M$ is strongly $\beta$-normalizing, then $\InferStrong(M)$ terminates.
The key idea for proving this result is to relate the behaviour of the algorithm to the \emph{reduction} of $M$;
this intuition can be strengthened by stating that $\InferStrong$ essentially behaves as a (non-deterministic) 
interpreter of $\l$-calculus, incrementally constructing a type for $M$ by reducing it. 

From this perspective, taking minimal pseudo-derivations as the starting
point is motivated by the fact that the minimal pseudo-derivation for $M$ matches the structure
of the term, \ie each subterm of $M$ is the subject of exactly one subderivation. 
During the execution of $\InferStrong$, the
structure of the pseudo-derivation evolves according to the replication of subterms
happening at reduction time: indeed, a blocked equation always originates from a \emph{non-linear}
redex that is created along a reduction sequence starting from $M$. Accordingly, normal
forms generate no blocked equation.

\begin{lemma} \label{lem:nf-no-expansion}
If $\PDsn{M}=(\Pi,\eqset)$ and $M$ is in normal form,
then $\eqset$ is solvable.
\end{lemma}

\begin{proof} 
By induction on the normal form $M$. 
We show that applying unification 
rules to $\eqset$ cannot generate equations between
lists, hence cannot generate blocked equations. 
Case $M=x$ is trivial, as $\eqset=\emptyset$. For the case $M=\lambda x.N$,
let $\PDsn{N} = (\Pi_0, \eqset_0)$; then 
$\eqset = \eqset_0 \cup \{\conc \doteq \ls \arrow \conb\}$ where $\conc$ is fresh. 
Since by \ih $\eqset_0$ does not generate equations between lists, 
neither does $\eqset$, because $\conc$ does not occur in $\eqset_0$.
Lastly, consider $M=xM_1 \dots M_n$ for $n>0$. 
Let $\PDsn{x}=(\Pi_0, \emptyset)$ where $\Pi_0 \dem x: \llist{\con_0} \derp x: \con_0$,
and $\PDsn{M_i}=(\Pi_i, \eqset_i)$ where $\Pi_i \dem \tGamma_i \derp M_i: \ls_i$ $(1 \leq i \leq n)$. Then 
$\eqset = \bigcup_{i=1}^n \eqset_i \cup \{\con_0 \doteq \ls_1 \arrow \con_1, \; \con_1 \doteq \ls_2 \arrow \con_2, \; \dots, \; \con_{n-1} \doteq \ls_n \arrow \con_n \}$ where $\con_1,\dots,\con_n$ are fresh. 
Remark that $\eqset_i * \eqset_j$ whenever $i \not = j$.
Since by \ih none of the $\eqset_i$ generates equations between lists, 
neither does $\eqset$, because $\con_0,\con_1,\dots,\con_n$ do not occur in $\bigcup_{i=1}^n \eqset_i$.
\end{proof}

\begin{corollary} \label{cor:nf-termination}
If $M$ is in normal form, then $\InferStrong(M)$ terminates performing no expansion.
\end{corollary}


The proof of the following \Cref{lem:algo-as-redb},
which relates the termination of $\InferStrong(M)$
to the termination of the procedure starting from reducts of $M$, 
employs Barendregt's $F_\infty$ \emph{perpetual reduction strategy}.
Let us recall that perpetual reduction strategies preserve infinite reduction paths \cite{RaamsdonkSSX99};
thanks to the close connections with strong normalization, 
such strategies have proven useful for studying properties 
of intersection type systems without empty intersection on several occasions.
In particular, we mention relevant work by Neergaard \cite{Neergaard05}, who 
provided an elegant inductive proof of (a weaker variant of) 
subject expansion in a system similar to $\systemITsn$, 
by considering expansion w.r.t. $F_\infty$-reduction steps.
An extension of such a proof technique is found in \cite{DudenhefnerP24}. 

\begin{definition}[{\cite[Definition 13.4.1]{Barendregt85book}}] \label{def:f-infinity}
The $F_\infty$ reduction strategy is defined as:
\[
F_\infty(M) =
\begin{cases}
\text{If $M$ is in normal form then $M$} \\
\text{If $M = \ccontext[(\l x.P)Q]$ and $(\l x.P)Q$ is the leftmost-outermost redex:} \\
\begin{cases}
\text{$\ccontext[P\subs{x}{Q}]$ if $x \in \FV(P)$} \\
\text{$\ccontext[P]$ if $x \not \in \FV(P)$ and $Q$ is in normal form} \\
\text{$\ccontext[(\l x.P)F_\infty(Q)]$ otherwise} \\
\end{cases}
\end{cases}
\]
\end{definition}


\begin{lemma}\label{lem:algo-as-redb}
Let $M = M_0 \redb M_1 \redb \dots \redb M_k$ be a $F_\infty$ reduction sequence 
(\ie let $M_{i+1} = F_\infty(M_{i})$ for $0 \leq i < k$).
Then $\InferStrong(M) $ terminates if and only if $\InferStrong(M_k) $ terminates. 
\end{lemma}

\begin{proof}
See \Cref{sec:termination-appendix}.
\end{proof}


Assuming \emph{fair} non-determinism, \ie 
that the algorithm does not repeat expansions with no effect an infinite number of times, 
we can state: 

\begin{theorem} \label{thm:termination}
If $M$ is strongly $\beta$-normalizing, then $\InferStrong(M)$ terminates.
\end{theorem}

\begin{proof}
By \Cref{lem:algo-as-redb} and \Cref{cor:nf-termination}.
\end{proof}

\begin{corollary} \label{cor:termination-iff-sn}
$\InferStrong(M)$ terminates if and only if $M$ is strongly $\beta$-normalizing.
\end{corollary}

\begin{proof}
($\Rightarrow$) By \Cref{thm:algo-correctness} and \Cref{thm:characterization}. ($\Leftarrow$) By \Cref{thm:termination}.
\end{proof}

\subsection{Confluence} \label{sec:confluence}

We proceed to show that, although the algorithm is non-deterministic, its output (if any) is unique. 
Before delving into details, some considerations about our definition of expansion are in order.


Expansion and erasure operations, as introduced by \Cref{def:expansion},
allow for a compact presentation, but are clearly suboptimal from an efficiency standpoint.
Indeed, it is possible to provide a more refined definition of expansion (resp. erasure), 
that modifies a given $\many$ rule by adding minimal premises 
(resp. by deleting premises) in \emph{specific} positions, instead of “starting over” from minimal subderivations only.
To reason about confluence properties of the algorithm,
it is convenient to assume this refined notion of expansion: 
we do so in stating the following \Cref{lem:algo-confluence} and \Cref{thm:algo-confluence}. 
This way we can rely on the fact that, 
during execution of $\InferStrong$, the length of all lists is \emph{non-decreasing}.

\vskip 1cm

\begin{lemma} \label{lem:algo-confluence}
Let $(\Pi,\eqset) \Rew (\Pi_1,\eqset_1)$ and $(\Pi,\eqset) \Rew (\Pi_2,\eqset_2)$. 
Then there is $(\Pi_3,\eqset_3)$
such that $(\Pi_1,\eqset_1) \Rew^* (\Pi_3,\eqset_3)$ 
and either $(\Pi_2,\eqset_2) \Rew (\Pi_3,\eqset_3)$ or $(\Pi_2,\eqset_2) = (\Pi_3,\eqset_3)$.
\end{lemma}

\begin{proof}
See \Cref{sec:confluence-appendix}.
\end{proof}

\Cref{lem:algo-confluence} implies that the relation $\Rew$ is confluent, 
hence that the output of $\InferStrong$ is uniquely determined. 
Given the similarities between our type inference procedure and term reduction,
this result can be seen as a natural consequence of the confluence properties of $\lambda$-calculus.

\begin{theorem} \label{thm:algo-confluence}
If $\InferStrong(M) = (\Pi,\psi)$ and $\InferStrong(M) = (\Pi',\psi')$, 
then $(\Pi,\psi) = (\Pi',\psi')$ modulo renaming of pre-type variables.
\end{theorem}

\begin{proof}
By \Cref{lem:algo-confluence} and \cite[Lemma 2.5]{Huet80},
plus \Cref{prop:unif-confluence} of unification.
\end{proof}

\subsection{Principality}

Let us recall that the principal typing for a term $M$ 
is a typing from which all other typings for $M$ can be derived
via suitable operations.
The procedure $\InferStrong(M)$, if it terminates, 
constructs the principal typing of the (strongly normalizing) term $M$, 
in the following sense.

\begin{theorem} \label{thm:algo-principal}
Let $\Sigma \dem \Gamma \dersn M: \iA$ and $\InferStrong(M) = (\Pi,\psi)$. 
Then $\Sigma = \ltom \circ \psi' \circ \expseq(\Pi)$ for some sequence of expansions $\expseq$
and substitution $\psi'$.
\end{theorem}

\begin{proof} 
By \Cref{lem:from-pd-to-der-strong}, every derivation $\Sigma$ in system $\systemITsn$ 
can be obtained from a suitable strong pseudo-derivation; 
moreover, the same Lemma guarantees that any strong pseudo-derivation can be 
built starting from $\PDsnmin{M}$ via a sequence of expansions. 
$\InferStrong$ starts from $\PDsnmin{M}$ and, if it terminates, 
by \Cref{thm:algo-confluence} its output is uniquely determined. 
Observe that the algorithm only performs expansions that are strictly required
to solve blocked equations, and does so in a \emph{minimal} way (\ie lists are never expanded more than needed).
We conclude that every $\Sigma$ can be obtained from $\Pi$ via a suitable sequence $\expseq$ of additional
(\ie not strictly necessary) expansions.  
\end{proof}

\begin{remark}
\Cref{thm:algo-principal} does not speak about the relationships 
between $\psi$ and $\psi'$. A formal analysis would 
be quite technical and goes beyond the scope of this paper. 
Here we only point out that $\dom{\psi} \subseteq \dom{\psi'}$,
as it is safe to assume $\Tvar(\Pi) \subseteq \Tvar(\expseq(\Pi))$.
Of course, even for $\con \in \dom{\psi} \cap \dom{\psi'}$ one may have $\psi(\con) \not = \psi'(\con)$, 
since in general lists in $\expseq(\Pi)$ are larger than those in $\Pi$.
Observe that if $\Sigma$ shares the same (minimal) tree structure as $\Pi$, then 
$\Sigma = \ltom \circ \psi' (\Pi)$ where $\psi' = \psi'' \circ \psi$ for some $\psi''$; 
this is because $\psi$ is the most general unifier of the set of equations associated to $\Pi$. 
\end{remark}

%% file: conclusion.tex

In this paper we 
design a type inference semi-algorithm for system $\systemITsn$, 
computing the principal typing  
of all and only the strongly $\beta$-normalizing terms 
(\Cref{cor:termination-iff-sn} and \Cref{thm:algo-principal}).
The procedure is based on essential notions that we prove correct and complete
w.r.t. $\systemIT$-typability (Theorems \ref{thm:correctness} and~\ref{thm:completeness}).
Although our results were already known, the methodology is new, and we leverage advancements 
in intersection type theory made over the past decades to streamline some technical details.

It is possible to see our semi-algorithm  
as a ``blueprint'' from which inference procedures for other type systems can be derived. 
For example, it is natural to consider a variation of $\InferStrong$ that,
making use of both expansion and erasure operations, constructs a derivation in system $\systemIT$ 
for all and only the $\beta$-normalizing terms.  
We point out that the interest of such a variation 
is somewhat limited, as $\systemIT$ also admits an alternative
(and conceptually simpler) type inference semi-algorithm:
first, reduce $M$ to normal form $N$ (if any, otherwise the procedure does not terminate); 
second, build a derivation for $N$ (easy by \Cref{lem:nf-no-expansion}); third, build a derivation 
for $M$ via subject expansion (\Cref{thm:subconv}). 
Of course this approach cannot be easily adapted to $\systemITsn$,
because subject expansion does not hold ``on the nose'' in strong systems (\Cref{rmk:strong-systems}).
Yet another variation on the theme yields the \emph{uniform} 
intersection type inference algorithm of \cite{PautassoR23}, 
which is nothing more than an \emph{always terminating} specialization of $\InferStrong$.
Additional details can be found in the first author's Ph.D. thesis \cite{Pautasso25}. 

Let us conclude by spending a few words about type inference
in \emph{idempotent} intersection type systems. We highlight that the approach presented here can be
applied also in an idempotent setting, as is. Indeed, when performing type inference, the most
general assumption one can make is that each copy of a subterm exhibits a different behaviour,
\ie must be assigned a different intersection type. Observe that we already followed this principle by
imposing disjointness conditions on pseudo-derivations (see \Cref{def:pseudo}): replacing multisets by sets would have
no real impact, because set union behaves exactly as multiset union when sets are disjoint.

%% file: appendix.tex


\section{Completeness} \label{sec:completeness-appendix}

We prove the two points of \Cref{lem:from-pd-to-der} separately, starting from:

\begin{enumerate}
\item \textit{Let $\PD{M}= (\Pi, \eqset)$ and $\PDmin{M}=(\Pi_M, \eqset_M)$. Then there is 
a sequence of expansions and erasures $\expseq$ such that 
$(\Pi, \eqset)=\expseq(\Pi_M, \eqset_M)$.}
\end{enumerate}

\begin{proof}
By induction on $\Pi$, considering its last rule.
\begin{itemize}
\item Case $\axvar$. It is immediate to check that 
$\PD{x}=\PDmin{x}$ modulo renaming of pre-type variables, hence $\expseq$ is the empty sequence.
\item Case $\abs$. Letting $M = \l x.N$, the pseudo-derivation $\Pi$ has shape:
\[
\infer{\Pi \dem \tGamma \derp \l x.N : \con}{\Pi_0 \dem \tGamma, x: \sigma \derp N : \conb}
\]
Letting $\PDmin{N}=(\Pi_N,\eqset_N)$, the minimal pseudo-derivation $\Pi_M$ has shape:
\[
\infer{\Pi_M \dem \tDelta \derp \l x.N : \con}{\Pi_N \dem \tDelta, x: \tau \derp N : \conb}
\]
By inductive hypothesis there is 
$\expseq_0$ such that $(\Pi_0,\eqset_0)=\expseq_0(\Pi_N, \eqset_N)$,
therefore $\expseq_0$ is the desired sequence.
\item Case $\app$. Letting $M = PQ$, the pseudo-derivation $\Pi$ has shape:
\[
\infer
{\Pi \dem \tGamma \derp PQ : \con}
{ \Pi_1 \dem \tGamma_1 \derp P : \conb & 
	\infer{\Pi_2 \dem \tGamma_2 \derp Q : \llist{\conc_1,...,\conc_n}}
	{(\Pi_{2i} \dem \tGamma_{2i} \derp Q : \conc_i)_{i=1}^n} 
}
\]
for some $n \geq 0$.
Letting $\PDmin{P} = (\Pi_P,\eqset_P)$ and $\PDmin{Q} = (\Pi_Q,\eqset_Q)$, 
the minimal pseudo-derivation $\Pi_M$ has shape:
\[
\infer
{\Pi_M \dem \tDelta \derp PQ : \con}
{ \Pi_P \dem \tDelta_1 \derp P : \conb & 
	\infer{\tDelta_{2} \derp Q : \llist{\conc_1}}
	{\Pi_Q \dem \tDelta_{2} \derp Q : \conc_1} 
}
\]
By inductive hypothesis there is $\expseq_1$ such that $(\Pi_1, \eqset_1)=\expseq_1(\Pi_P, \eqset_P)$.
Hence, if $n = 0$, the desired sequence is $\expseq_1 \circ \Erase{\llist{\conc_1},1}$.
Otherwise, if $n \geq 1$, consider also $n$ disjoint copies of $(\Pi_Q, \eqset_Q)$, namely
$(\Pi^i_Q,\eqset^i_Q)$ for $1 \leq i \leq n$. 
By inductive hypothesis, for each $i$ there is $\expseq_{2i}$ 
such that $(\Pi_{2i}, \eqset_{2i})=\expseq_{2i}(\Pi^i_Q, \eqset^i_Q)$.
Let $\expseq'$ be a sequence obtained by composing, 
in any order, the sequences $\expseq_1, \expseq_{21},\dots, \expseq_{2n}$ 
(observe that, since $\expseq_1, \expseq_{21},\dots, \expseq_{2n}$ 
all act on different subtrees, the order of composition is not important).
If $n=1$, the desired sequence is simply $\expseq'$.
Lastly, if $n \geq 2$, the desired sequence is $\expseq' \circ \Expand{\llist{\conc_1},n-1}$.
\end{itemize}
\end{proof}



\noindent Now we prove the second point of \Cref{lem:from-pd-to-der}, namely:

\begin{enumerate}
\setcounter{enumi}{1}
\item \textit{Let $\Sigma \dem \Gamma \der M:\iA$. Then there are a pseudo-derivation $\PD{M}=(\Pi, \eqset)$ 
and a substitution $\psi$ such that $\psi$ solves $\eqset$ and $\Sigma = \ltom \circ \psi(\Pi)$.}
\end{enumerate}

\begin{proof}
We show that if $\Pi$ shares the same tree structure as $\Sigma$,
then the desired $\psi$ is easily found. We proceed by induction on $\Sigma$, considering its last rule. 
In what follows, $\psi_1 \cup \psi_2$ denotes the union of compatible substitutions.
Moreover, for brevity, we simply write $\cod{\psi}$ instead of $\Tvar(\cod{\psi})$.

\begin{itemize}
\item Case $\axvar$. Letting $M = x$, the derivation $\Sigma$ has shape:
\[
\infer{\Sigma \dem x : \mult{\iA} \derp x : \iA}{}
\]
As necessarily $\PD{M} = (\Pi \dem x : \llist{\con} \derp x : \con, \emptyset)$, 
any $\psi$ such that $\psi(\con) = \A$ and $\ltom(\A) = \iA$ meets the requirements.
\item Case $\abs$. Letting $M = \l x.N$ and $\iA = \ms \arrow \iB$, the derivation $\Sigma$ has shape:
\[
\infer{\Sigma \dem \Gamma \der \l x.N : \ms \arrow \iB}{\Sigma_0 \dem \Gamma, x:\ms \der N : \iB}
\]
By inductive hypothesis there is 
$\PD{N} = (\Pi_0, \eqset_0)$ and  $\psi_0$ solving $\eqset_0$ such that
$\Sigma_0 = \ltom \circ \psi_0(\Pi_0)$.
Now consider the pseudo-derivation $\PD{M} = (\Pi, \eqset)$ of shape:
\[
\infer{\Pi \dem \tGamma \derp \l x.N : \con}{\Pi_0 \dem \tGamma, x: \ls \derp N : \conb }
\qquad 
\eqset = \eqset_0 \cup \set{\con \doteq \ls \arrow \conb}
\]
Since $\con$ is fresh, it is safe to assume $\con \not \in \dom{\psi_0} \cup \cod{\psi_0}$; 
therefore the substitution $\psi = \psi_0 \cup \{\con \doteq \psi_0(\ls \arrow \conb)\}$ meets the requirements.
\item Case $\app$. Letting $M = PQ$, the derivation $\Sigma$ has shape:
\[
\infer
{\Sigma \dem \Gamma \der PQ : \iA}
{ \Sigma_1 \dem \Gamma_1 \der P : \ms \arrow \iA & 
	\infer{\Sigma_2 \dem \Gamma_2 \der Q : \ms}
	{(\Sigma_{2i} \dem \Gamma_{2i} \der Q : c_i)_{i=1}^n} 
}
\]
where $n \geq 0$. 
By inductive hypothesis, there are $\PD{P} = (\Pi_1,\eqset_1)$ and $\psi_1$ solving $\eqset_1$ such that
$ \Sigma_1 = \ltom \circ \psi_1(\Pi_1)$.
Moreover, if $n > 0$, by inductive hypothesis there are 
$\PD{Q} = (\Pi_{2i}, \eqset_{2i})$ and
$\psi_{2i}$ solving $\eqset_{2i}$ for each $1 \leq i \leq n$, such that
$\Sigma_{2i} = \ltom \circ \psi_{2i}(\Pi_{2i})$.
Since $\Pi_{2j} * \Pi_{2k}$ whenever $j \not = k$, 
it is safe to assume that $\psi_{21},\dots,\psi_{2n}$ are such that  
$\dom{\psi_{2j}} \cap \dom{\psi_{2k}} = \dom{\psi_{2j}} \cap \cod{\psi_{2k}} = \emptyset $ 
whenever $j \not = k$; therefore $\psi_2 = \bigcup_{i=1}^n \psi_{2i}$ is a substitution
solving $\eqset_{2} = \bigcup_{i=1}^n \eqset_{2i}$.
Now consider the pseudo-derivation $\PD{M} = (\Pi,\eqset)$ of shape: 
\[
\infer[]
{\Pi \dem \tGamma \derp PQ : \con}
{ \Pi_1 \dem \tGamma_1 \derp P : \conb & 
  \infer{\Pi_2 \dem \tGamma_2 \derp Q : \ls}{(\Pi_{2i} \dem \tGamma_{2i} \derp Q : \conc_i)_{i=1}^n}
}
\quad
\eqset = \eqset_1 \cup \eqset_2 \cup \{\conb \doteq \sigma \arrow \con\}
\]
Note that in case $n=0$ one has $\ls = \ellist$ and $\eqset_2 = \emptyset$.
Again, as $\Pi_1 * \Pi_2$, it is safe to assume that for $j,k \in \{1,2\}$ we have 
$ \dom{\psi_j} \cap \dom{\psi_k} = \dom{\psi_j} \cap \cod{\psi_k} = \emptyset $
whenever $j \not = k$; therefore $\psi_1 \cup \psi_2$ is a substitution.
Since $\ltom \circ \psi_1 (\conb) = \ms \arrow \iA$, 
it must be the case that $\psi_1(\conb) = \lstwo \arrow \A$ for some $\lstwo,\A$
such that $\ltom(\lstwo) = \ms = \ltom \circ \psi_2(\ls)$ and $\ltom(\A) = \iA$.
As the lists $\lstwo$ and $\psi_2(\ls)$ collapse into the same multiset $\ms$,
they are the same up to permutation; 
since it is always possible to reorder the $n$ premises of $\Pi_2$, 
w.l.o.g. we can assume $\lstwo = \psi_2(\ls)$.
Lastly, since $\con$ is fresh, we can assume 
$\con \not \in \dom{\psi_1} \cup \dom{\psi_2} \cup \cod{\psi_1} \cup \cod{\psi_2}$;
therefore the substitution $\psi = \psi_1 \cup \psi_2 \cup \{\con \doteq \A\}$ meets the requirements.
\end{itemize}
\end{proof}


\section{Termination} \label{sec:termination-appendix}

\input{termination}


\section{Confluence} \label{sec:confluence-appendix}

\begin{lemma*} [\ref{lem:algo-confluence}]
Let $(\Pi,\eqset) \Rew (\Pi_1,\eqset_1)$ and $(\Pi,\eqset) \Rew (\Pi_2,\eqset_2)$. 
Then there exists $(\Pi_3,\eqset_3)$
such that $(\Pi_1,\eqset_1) \Rew^* (\Pi_3,\eqset_3)$ 
and either $(\Pi_2,\eqset_2) \Rew (\Pi_3,\eqset_3)$ or $(\Pi_2,\eqset_2) = (\Pi_3,\eqset_3)$.
\end{lemma*}

\begin{proof}
(Sketch) We analyse the cases where two coinitial expansions
that the algorithm can perform may influence each other.
Let $\Pi$ contain a subderivation of shape:
\[
\infer[\app]{ \tGamma \derp PQ : \con}{ \Sigma_1 \dem \tGamma_1 \derp P : \conb & \Sigma_2 \dem \tGamma_2 \derp Q : \ls}
\]
and assume that applying unification rules to $\eqset$ results in a blocked equation $\lstwo \doteq \ls$,
originated from an equation $\conb \doteq \lstwo \arrow \A$ for some $\lstwo$ such that $|\lstwo| \not = |\ls|$.
Let $\expstep_2$ be the expansion that the algorithm can perform to try to unblock such an equation.
We distinguish two cases:
\begin{itemize}
\item First, consider the case where,
in addition to $\expstep_2$, the algorithm can also perform
an expansion $\expstep_1$ inside $\Sigma_1$.
Then $\Pi \Rew \Pi_1$ and $\Pi \Rew \Pi_2$, 
where $\Pi_1 = \expstep_1(\Pi)$ and $\Pi_2 = \expstep_2(\Pi)$.

Let us start from the scenario $|\ls| < |\lstwo|$.
If in $\Pi_1$ the list $\lstwo$ is replaced by $\lstwo'$ such that |$\lstwo| < |\lstwo'|$,
the algorithm can no longer apply $\expstep_2$;
instead, the algorithm can build $\Pi_3 = \Expand{\ls,|\lstwo'|-|\ls|,\Pi_1}$.
In $\Pi_2$, the list $\ls$ is replaced by $\ls'$ such that $|\ls'| = |\lstwo|$, 
and $\expseq_1$ is still applicable.
Therefore, starting from $\Pi' = \expstep_1(\Pi_2)$,
the algorithm can build the same pseudo-derivation as before, 
namely $\Pi_3 = \Expand{\ls',|\lstwo'|-|\ls'|,\Pi'}$.
Thus we have $\Pi_1 \Rew \Pi_3$ and $\Pi_2 \Rew^* \Pi_3$.
If, on the other hand, $\expstep_1$ does not modify $\lstwo$,
one can check that $\expstep_2(\Pi_1) = \expstep_1(\Pi_2)= \Pi_3$,
\ie $\Pi_1 \Rew \Pi_3$ and $\Pi_2 \Rew \Pi_3$.

Now observe that, during
the whole execution of $\InferStrong(M)$, the scenario $|\lstwo| < |\ls|$ cannot happen: indeed, the length of
the list occurring in the major premise of an $\app$ rule is always greater or equal
than the length of the list occurring in its minor premise. 
To see why, note that this property holds for
the starting point of the algorithm, namely the minimal strong pseudo-derivation $\PDsnmin{M}$
(since $n=1$ in all rules $\many$), and is preserved by expansion operations (recall that we assume a definition
of expansion such that the length of lists is non-decreasing).
\item Second, consider the case where,
in addition to $\expstep_2$, the algorithm can also perform
an expansion $\expstep_1$ inside $\Sigma_2$.
One can check that $\expstep_1$ does not modify $\ls$ nor $\lstwo$, 
hence $\expstep_2$ can be applied to $\Pi_1 = \expstep_1(\Pi)$.
Similarly, since $\expstep_2$ introduces additional fresh copies of $\PDsnmin{Q}$
without interfering with the existing premises of $\Sigma_2$
(including the one on which $\expstep_1$ acts), the expansion
$\expstep_1$ can be applied to $\Pi_2 = \expstep_2(\Pi)$.
Therefore $\expstep_2 \circ \expstep_1(\Pi) = \expstep_1 \circ \expstep_2(\Pi) = \Pi_3$,
\ie  $\Pi_1 \Rew \Pi_3$ and $\Pi_2 \Rew \Pi_3$.
\end{itemize}
\end{proof}

%% file: termination.tex


\begin{lemma*}[\ref{lem:algo-as-redb}]
Let $M = M_0 \redb M_1 \redb \dots \redb M_k$ be a $F_\infty$ reduction sequence 
(\ie let $M_{i+1} = F_\infty(M_{i})$ for $0 \leq i < k$).
Then $\InferStrong(M) $ terminates if and only if $\InferStrong(M_k) $ terminates. 
\end{lemma*}

\begin{proof} 
In what follows, 
we write $\uniseq(\eqset)$ for the result of applying 
a sequence of unification rules $\uniseq$ to the set of equations $\eqset$.
The union of \emph{disjoint sets} is noted $\sqcup$, 
\ie writing $X \sqcup Y$ tacitly means that $X \cap Y = \emptyset$.
We argue by induction on $k$.

\paragraph{Case $k = 1$.}
We develop in details the base step, 
namely $M = \ccontext[(\lambda x.P)Q] \redb \ccontext[P\subs{x}{Q}] = M_1 = N$.
Let $\PDsnmin{M}=(\Pi_M,\eqset_M)$ and $\PDsnmin{N}=(\Pi_N,\eqset_N)$.  
The proof depends on the occurrences of $x$ in $P$. 

\paragraph*{Case $x$ occurs in $P$.}
We show that $\InferStrong(M)$ can transform $\eqset_M$ into a superset $\eqset_M^\dagger$ of $\eqset_N$, 
that is $\eqset_N \subset \eqset_M^\dagger$;
at the same time we show that equations in $\eqset_M^\dagger - \eqset_N$ cannot
play any role in generating blocked equations when solving $\eqset_M^\dagger$, \ie that
$\eqset_M^\dagger$ is solvable if and only if $\eqset_N$ is.
To ease the presentation, for the time being we assume $P \not = x$; 
the special case $P = x$ is treated separately later on. 
The derivation $\Pi_M$ contains a subderivation $\Pi_{(\lambda x.P)Q}$ of shape:
\[ 
\infer{\Pi_{(\lambda x.P)Q} \dem \tGamma \cons \tDelta \derp (\lambda x.P)Q: \cone}
{\infer[\absi]{ \tGamma \derp \lambda x.P: \conc}{\Pi_P \dem \tGamma, x: \llist{\con_i}_{i \in I} \derp P: \conb  & \conc \text{ fresh} } 
& \infer{\tDelta \derp Q: \llist{\cond} }{\Pi_Q \dem \tDelta \derp Q: \cond} & \cone \text{ fresh}} 
\]
where $|I| \geq 1$, $\PDsnmin{P} = (\Pi_P,\eqset_P)$ and  $\PDsnmin{Q}=(\Pi_Q,\eqset_Q)$.
By construction, for each occurrence of $x$ in $P$ there is an axiom with subject $x$ in $\Pi_P$.
Since we assumed $P \not = x$, in $\Pi_P$ each such axiom is followed either by a rule $\abs$ or by a rule $\app$. 
Let us consider all possible scenarios by partitioning the set of indexes $I$ into three subsets, 
\ie letting $I = I_\eabs \sqcup I_\eappl \sqcup I_\eappr$.
For each axiom with subject $x$ followed by an $\abs$ rule, $\Pi_P$ contains a subderivation $\Xi_i$ of shape:
\[
\vcenter{ \infer{ \Xi_i \dem x :\llist{\con_i} \derp \lambda y.x : \conf_i}
{x :\llist{\con_i} \derp x: \con_i & \conf_i, \conf_i' \text{ fresh}} } \qquad (i \in I_\eabs)  
\] 
where $y \not = x$.
Let the set of equations generated by all the aforementioned $\Xi_i$ be:
\[ \eqset_{\eabs} = \{ \conf_i \doteq \llist{\conf_i'} \arrow \con_i  \mid  i \in I_{\eabs} \} \]
On the other hand, if the axiom introducing $x$ is followed by an $\app$ rule we distinguish two subcases, 
depending on whether $x$ occurs on the left or on the right side of the application.
For each occurrence of $x$ in $P$ in functional position,  
$\Pi_P$ contains a subderivation $\Sigma_i$ of shape:
\[
\vcenter{ \infer{ \Sigma_i \dem \{x:\llist{\con_i}\} \cons \tPsi_i \derp xR_i: \conh_i}
{x:\llist{\con_i} \derp x: \con_i 
& \infer{\tPsi_i \derp R_i: \llist{\cong_i}}{\tPsi_i \derp R_i: \cong_i} 
& \conh_i \text{ fresh}} } \qquad (i \in I_\eappl)  
\] 
Let the set of equations generated by the last rule of all the aforementioned $\Sigma_i$ be:
\[ \eqset_{\eappl} = \{ \con_i \doteq \llist{\cong_i} \arrow \conh_i  \mid  i \in I_{\eappl} \} \]
Focusing instead on the occurrences of $x$ in $P$ in argument position, 
$\Pi_P$ contains subderivations $\Theta_i$ of shape:
\[
\vcenter{ \infer{ \Theta_i \dem \tPhi_i \cons \{x:\llist{\con_i}\} \derp S_i x: \conl_i}
{ \tPhi_i \derp S_i: \coni_i 
& \infer{x:\llist{\con_i} \derp x:\llist{\con_i}}{x:\llist{\con_i} \derp x:\con_i} 
& \conl_i \text{ fresh}} } \qquad (i \in I_\eappr) 
\] 
Let the set of equations generated by the last rule of all the aforementioned $\Theta_i$ be:
\[ \eqset_{\eappr} = \{ \coni_i \doteq \llist{\con_i} \arrow \conl_i \mid i \in I_{\eappr} \} \]
Let $\eqset_P^x = \eqset_\eabs \cup \eqset_\eappl \cup \eqset_\eappr$ 
(observe that clearly $\eqset_P^x \subseteq \eqset_P$); 
moreover, let $\eqset_P^* = \eqset_P - \eqset_P^x$ be the subset of constraints generated by $\Pi_P$
which are not related to the occurrences of $x$ in $P$.
From the above considerations, it follows that $\eqset_M$ has shape:
\[ \eqset_M = \eqset \sqcup \eqset_P^* \sqcup \eqset_P^x \sqcup \eqset_Q \sqcup \{\conc \doteq \llist{\con_i}_{i\in I} \arrow \conb, \conc \doteq \llist{\cond} \arrow \cone\} \] 
for some $\eqset$.
If $|I| = n > 1$, \ie $(\l x.P)Q$ is a non-linear redex, 
applying unification rules to $\eqset_M$ yields a blocked equation 
$\llist{\con_1,\dots,\con_n} \doteq \llist{\cond}$ 
(for the case $|I| = n = 1$, ignore the following considerations about the expansion and proceed directly
to the discussion of $\uniseq$, letting $\Pi_M' = \Pi_M$ and $\eqset_M' = \eqset_M$).
At this point, $\InferStrong$ can perform an expansion, obtaining 
a pseudo-derivation $\Expand{\llist{\cond}, n-1, \Pi_M} =(\Pi_M',\eqset_M')$.
In place of $\Pi_{(\lambda x.P)Q}$, the derivation $\Pi_M'$ contains 
a subderivation $\Pi_{(\lambda x.P)Q}'$ of shape:
\[ 
\infer{\Pi_{(\lambda x.P)Q}' \dem \tGamma \cons_{i \in I} \tDelta_i \derp (\lambda x.P)Q: \cone}
{\infer{ \tGamma \derp \lambda x.P: \conc}{\Pi_P \dem \tGamma, x: \llist{\con_i}_{i \in I} \derp P: \conb  & \conc \text{ fresh} } 
& \infer{\cons_{i \in I} \tDelta_i \derp Q: \llist{\cond_i}_{i \in I}}{(\Pi_Q^i \dem \tDelta_i \derp Q: \cond_i)_{i \in I}} 
& \cone \text{ fresh}} 
\]
where $(\Pi_Q^i, \eqset_Q^i)$ ($i \in I$) are $n$ disjoint copies of $\PDsnmin{Q}$. 
Hence we have:
\[ \textstyle
\eqset_M' = \eqset' \sqcup \eqset_P^* \sqcup \eqset_P^x \sqcup (\bigsqcup_{i \in I} \eqset_Q^i)
\sqcup \{\conc \doteq \llist{\con_i}_{i\in I} \arrow \conb,\conc \doteq \llist{\cond_i}_{i \in I} \arrow \cone\} 
\] 
for some $\eqset'$. Remark that $\con_i$ ($i \in I$), $\conc$ do not occur in 
$\eqset' \sqcup \eqset_P^* \sqcup (\bigsqcup_{i \in I} \eqset_Q^i)$, 
and that $\conc, \cone$ do not occur in $\eqset_P \sqcup (\bigsqcup_{i \in I} \eqset_Q^i)$.
Applying a suitable sequence $\uniseq$ of unification rules to $\eqset_M'$,
the equation $\llist{\con_i}_{i \in I} \doteq \llist{\cond_i}_{i \in I}$ can be decomposed, 
each $\con_i$ can be substituted by $\cond_i$, and $\cone$ can be substituted by $\conb$.
This yields:
\begin{gather*}
\uniseq(\eqset_{\eabs}) = \{ \conf_i \doteq \llist{\conf_i'} \arrow \cond_i  \mid  i \in I_{\eabs} \}  \\[5pt]
\uniseq(\eqset_{\eappl}) = \{ \cond_i \doteq \llist{\cong_i} \arrow \conh_i  \mid  i \in I_{\eappl} \}  \\[5pt]
\uniseq(\eqset_{\eappr}) = \{ \coni_i \doteq \llist{\cond_i} \arrow \conl_i \mid i \in I_{\eappr} \}  \\[5pt]
\textstyle
\uniseq(\eqset_M')  =  
\eqset'\subs{\cone}{\conb} \sqcup \eqset_P^* \sqcup \uniseq(\eqset_P^x) \sqcup (\bigsqcup_{i \in I} \eqset_Q^i) \sqcup \{\conc \doteq \llist{\cond_i}_{i \in I} \arrow \conb, \cone \doteq \conb\} \sqcup \{\con_i \doteq \cond_i\}_{i \in I}
\end{gather*}
where $\uniseq(\eqset_P^x) = \uniseq(\eqset_\eabs) \cup \uniseq(\eqset_\eappl) \cup \uniseq(\eqset_\eappr)$. 
Now we show that: 
\begin{equation} \label{eq:constraints-relation-two}
\textstyle
\eqset_N = \eqset'\subs{\cone}{\conb} \sqcup \eqset_P^* \sqcup \uniseq(\eqset_P^x) \sqcup (\bigsqcup_{i \in I} \eqset_Q^i)
\end{equation} 
We start by pointing out that $\Pi_N$, in place of $\Pi_{(\lambda x.P)Q}$, 
contains a subderivation $\PDsnmin{P\subs{x}{Q}}$ $= (\Pi_{P\subs{x}{Q}}, \eqset_{P\subs{x}{Q}})$ ending by: 
\[ \Pi_{P\subs{x}{Q}} \dem \tGamma \cons_{i \in I} \tDelta_i \derp P\subs{x}{Q} : \conb \]
In turn, $\Pi_{P\subs{x}{Q}}$ contains subderivations corresponding to 
the various $\Xi_i$, $\Sigma_i$ and $\Theta_i$, namely:
\[
\vcenter{ \infer{ \Xi^{\subs{x}{Q}}_i \dem \tDelta_i \derp \lambda y.Q : \conf_i}
{ \Pi_Q^i \dem \tDelta_i \derp Q: \cond_i & \conf_i, \conf_i' \text{ fresh}} } \qquad (i \in I_\eabs)  
\]  
\[
\vcenter{ \infer{ \Sigma^{\subs{x}{Q}}_i \dem \tDelta_i \cons \tPsi_i \derp Q(R_i\subs{x}{Q}): \conh_i }
{\Pi_Q^i \dem \tDelta_i \derp Q: \cond_i & 
\infer{\tPsi_i \derp R_i\subs{x}{Q}: \llist{\cong_i}}{\tPsi_i \derp R_i\subs{x}{Q}: \cong_i} & 
\conh \text{ fresh}} }  
\qquad (i \in I_\eappl) 
\]
\[
\vcenter{ \infer{ \Theta^{\subs{x}{Q}}_i \dem \tPhi_i \cons \tDelta_i \derp (S_i\subs{x}{Q})Q: \conl_i }
{ \tPhi_i \derp S_i\subs{x}{Q}: \coni_i & 
\infer{\tDelta_i \derp Q: \llist{\cond_i}}{\Pi_Q^i \dem \tDelta_i \derp Q: \cond_i} & 
\conl \text{ fresh}} } 
\qquad (i \in I_\eappr)
\]
where $(\Pi^i_Q, \eqset_Q^i)$ ($i \in I$) are $n$ disjoint copies of $\PDsnmin{Q}$ 
(exactly as in the right premise of $\Pi_{(\lambda x.P)Q}'$).
Observe that the environments of
$\Xi^{\subs{x}{Q}}_i$, $\Sigma^{\subs{x}{Q}}_i$ and $\Theta^{\subs{x}{Q}}_i$ 
are possibly larger than those of 
$\Xi_i$, $\Sigma_i$ and $\Theta_i$ ($i \in I$). 
Therefore one may wonder if, due to such modifications,
the $\abs$ rules of $\Pi_N$ could generate equations that are different 
from those in $\uniseq(\eqset_M')$. 
We prove this is not the case.
First, notice that bound variables in $P$ cannot belong to $\dom{\tDelta_i} = \dom{\tDelta}$. 
Consequently, the equations generated by $\Pi_{P\subs{x}{Q}}$ are unaffected
by the enlargement of the environments, \ie we have 
$\eqset_{P\subs{x}{Q}} = \eqset_P^* \sqcup \uniseq(E_P^x) \sqcup (\bigsqcup_{i \in I} \eqset_Q^i)$. 
It remains to show that $\eqset_N - \eqset_{P\subs{x}{Q}} = \eqset'\subs{\cone}{\conb}$.
For most equations in $\eqset_N - \eqset_{P\subs{x}{Q}}$ this is easy, 
but once again special care is required for 
the equations generated by $\abs$ rules of $\Pi_N$ 
abstracting the various $z \in \dom{\tDelta}$. 
Consider any such $z$, and let $p$ (resp. $q$) be the number of 
axioms with subject $z$ in $\Pi_P$ (resp. $\Pi_Q$).
By construction, there are $m = h + p + q$ axioms with subject $z$ in $\Pi_M$, for some $h \geq 0$.
On the other hand, there are $m' = h + p + n \times q$ 
axioms with subject $z$ in $\Pi_N$. Looking at $\Pi_{(\lambda x.P)Q}'$, 
one realises that also $\Pi_M'$ contains exactly $m'$ axioms with subject $z$.
Thus, the rule $\abs$ abstracting $z$ in $\Pi_N$
generates an equation that belongs to $\eqset'\subs{\cone}{\conb}$.
Having proved \Cref{eq:constraints-relation-two}, we conclude by observing that 
$\uniseq(\eqset_M') - \eqset_N = \{\conc \doteq \llist{\cond_i}_{i \in I} \arrow \conb, \cone \doteq \conb\} 
\sqcup \{\con_i \doteq \cond_i\}_{i \in I}$ 
cannot play any role in generating blocked equations when solving $\uniseq(\eqset_M')$, 
because $\con_i$ $(i \in I)$, $\conc$ and $\cone$ do not occur in $\eqset_N$.
\paragraph*{The special case $P = x$.}
The derivation $\Pi_M$ contains a subderivation $\Pi_{(\lambda x.x)Q}$ of shape:
\[ 
\infer{\Pi_{(\lambda x.x)Q} \dem \tDelta \derp (\lambda x.x)Q: \cond}
{\infer[\absi]{ \derp \lambda x.x: \conb}{\Pi_P \dem  x: \llist{\con} \derp x: \con  & \conb \text{ fresh} } 
& \infer{\tDelta \derp Q: \llist{\conc} }{\Pi_Q \dem \tDelta \derp Q: \conc} & \cond \text{ fresh}} 
\]
Hence we have:
\[ \eqset_M = \eqset \sqcup \eqset_Q \sqcup \{\conb \doteq \llist{\con} \arrow \con, \conb \doteq \llist{\conc} \arrow \cond\} \]
for some $\eqset$.  Remark that $\con$, $\conb$ do not occur in $\eqset \sqcup \eqset_Q$, 
and that $\cond$ does not occur in $\eqset_Q$.  
Applying a suitable sequence $\uniseq$ of unification rules to $\eqset_M$, one obtains:
\[  \uniseq(\eqset_M) = \eqset\subs{\cond}{\conc} \sqcup \eqset_Q \sqcup \{\conb \doteq \llist{\conc} \arrow \conc, \con \doteq \conc, \cond \doteq \conc\} \]
The reader can easily check that $\eqset_N = \eqset\subs{\cond}{\conc} \sqcup \eqset_Q$;
indeed, in $\Pi_N$ the subderivation corresponding to $\Pi_{(\lambda x.x)Q}$ is simply $\Pi_Q$.
We conclude by observing that $\con$, $\conb$ and $\cond$ do not occur in $\eqset_N$; therefore 
$\uniseq(\eqset_M) - \eqset_N = \{\conb \doteq \llist{\conc} \arrow \conc, \con \doteq \conc, \cond \doteq \conc\}$
cannot contribute to the generation of blocked equations when solving $\uniseq(\eqset_M)$.

\paragraph*{Case $x$ does not occur in $P$.}
Let us start by highlighting that the $F_\infty$ reduction strategy selects 
the leftmost-outermost redex in a suitable subterm (see \Cref{def:f-infinity}).
Recall that a \emph{neutral term} has shape $x N_1 \dots N_n$ ($n \geq 0$), 
where each $N_i$ is a normal form, 
and that the leftmost-outermost reduction strategy $\redlo$
can reduce under the following contexts:
\begin{enumerate}[(a)]
\item \label{lo-abs} if $M \redlo M'$, then $\l x. M \redlo \l x.M'$;
\item \label{lo-appl} if $M \redlo M'$, then $MN \redlo M'N$;
\item \label{lo-appr} if $M \redlo M'$ and $N$ is a neutral term, then $NM \redlo NM'$.
\end{enumerate}
Now consider the derivation $\Pi_M$, which contains a subderivation $\Pi_{(\lambda x.P)Q}$ of shape:
\[ 
\infer{\Pi_{(\lambda x.P)Q} \dem \tGamma \cons \tDelta \derp (\lambda x.P)Q: \cond}
{\infer[\absk]{ \tGamma \derp \lambda x.P: \conb}
	{\Pi_P \dem \tGamma \derp P: \con  & \conb, \cone \text{ fresh} } 
& \infer{\tDelta \derp Q: \llist{\conc} }{\Pi_Q \dem \tDelta \derp Q: \conc} & \cond \text{ fresh}} 
\]
where $\PDsnmin{P} = (\Pi_P,\eqset_P)$ and $\PDsnmin{Q}=(\Pi_Q,\eqset_Q)$. Consequently $\eqset_M$ has shape:
\[\eqset_M = \eqset \sqcup \eqset_P \sqcup \eqset_Q \sqcup \{\conb \doteq \llist{\cone} \arrow \con, \conb \doteq \llist{\conc} \arrow \cond\}\] 
for some $\eqset$. 
First, we point out that in $\Pi_N$ the subderivation 
corresponding to $\Pi_{(\lambda x.P)Q}$ is simply $\Pi_P$,
hence $\eqset_N$ does not contain $\eqset_Q$ nor $\{\conb \doteq \llist{\cone} \arrow \con, \conb \doteq \llist{\conc} \arrow \cond\}$. Clearly the latter set does 
not contribute to generating blocked equations, as $\conb$, $\cond$ and $\cone$ are fresh.
Similarly, one can check that
the equations in $\eqset_Q$ can essentially be ignored: 
indeed, by \Cref{lem:nf-no-expansion}, 
solving $\eqset_Q$ \emph{by itself} yields a solved form, because $Q$ is in normal form.  
However, if some free variables of $Q$
are abstracted, 
pre-types variables occurring in $\eqset_Q$ may also occur in $\eqset$;
we prove that such variables do not play any significant role in the solution of $\eqset_M$, 
hence that $\eqset_M$ is solvable if and only if $\eqset_N$ is. 
Let $y \in \FV(Q)$ be an abstracted variable, and let $\tDelta (y) = \ls$
(note that, of course, pre-type variables in $\ls$ do not occur in $\eqset_N$). 
Now consider the subterm $\l y.R$ of $M$; remark that $R$ contains $(\l x.P)Q$. 
We analyse the two cases that may originate equations involving pre-type variables in $\ls$;
this happens when a term containing $\l y.R$ becomes either 
the functional or the argument part of an application.
For readability, in what follows we slightly abuse the syntax of pseudo-derivations, 
by implicitly performing some trivial unification steps. 
\begin{enumerate}[(i)]
\item \label{appl-eq}
In the first case, we reason by contradiction. 
Assume there is a subderivation of $\Pi_M$ of shape:
\begin{prooftree}
\def\ScoreOverhang{0pt}
\AxiomC{$\tPhi \derp \l y.R : \ls_1 \cons \ls \cons \ls_2 \arrow \conf$}
\UnaryInfC{$\vdots$}
\UnaryInfC{$\tPsi_1 \derp \kcontext[\l y.R] : \ls_1 \cons \ls \cons \ls_2 \arrow \conf$}	
	\AxiomC{$\tPsi_2 \derp S : \lstwo $}
\BinaryInfC{$\tPsi \derp \kcontext[\l y.R] S : \conf$}
\end{prooftree}
for some reduction context $\kcontext$. By definition of $F_\infty$, 
together with points (\ref{lo-abs}) and (\ref{lo-appl}) of $\redlo$-reduction,
the context $\kcontext$ must be generated by the grammar:
\[  \kcontext \grameq \square \mid \l x.\kcontext \mid \kcontext M  \]
However, this would create at least one redex containing $(\l x.P)Q$ as a subterm, 
thus contradicting the assumption that $(\l x.P)Q$ was the leftmost-outermost redex in $\kcontext[\l y.R] S$.
Therefore, this scenario is not possible.
\item \label{appr-eq} 
For the second case, consider a subderivation of $\Pi_M$ of shape:
\begin{prooftree}
\AxiomC{$\tPsi_1 \derp S : \cong$}
	\AxiomC{$\tPhi \derp \l y.R : \ls_1 \cons \ls \cons \ls_2 \arrow \conf$}
	\UnaryInfC{$\vdots$}
	\UnaryInfC{$\tPsi_2 \derp \l v_1 \dots v_m y.R : \A $}
	\UnaryInfC{$\tPsi_2 \derp \l v_1 \dots v_m y.R : \llist{\A} $}
\BinaryInfC{$\tPsi \derp S(\l v_1 \dots v_m y.R) : \conh$}
\end{prooftree} 
where $\A = \lstwo_1 \arrow \dots \arrow \lstwo_m \arrow \ls_1 \cons \ls \cons \ls_2 \arrow \conf$ for some $m \geq 0$.
By definition of $F_\infty$ and point (\ref{lo-appr}) of $\redlo$-reduction, 
$S$ must be a neutral term, \ie $S = z S_1 \dots S_n$ for some $n \geq 0$.
Let $z : \llist{\conl} \derp z : \conl$ be the axiom introducing $z$; 
then unification rules generate an equation of shape
$\conl \doteq \tau_1' \arrow \dots \arrow \lstwo_n' \arrow \llist{\A} \arrow \conh $.
A potentially problematic scenario could arise if some other equation involving $\conl$ were
generated, because then the pre-types variables contained in
$\ls$ could play a role in the solution of $\eqset_M$.
However, as seen before, this can only happen if $z$ is abstracted, and later a subterm containing $\l z.T$ 
becomes either the functional or the argument part of an application. 
Observe that the reasoning discussed in points (\ref{appl-eq}) and (\ref{appr-eq})
can be repeated replacing $\l y.R$ by $\l z.T$;
this would lead to a variable $z'$ in head position, an axiom $z' : \llist{\conl'} \derp z' : \conl'$,
an equation $\conl' \doteq \B$ such that $\conl$ occurs in $\B$, and so on.
Since this process cannot be iterated ad infinitum, we conclude
that only a finite chain of equations that do not impact the solvability of $\eqset_M$
can be created.
\end{enumerate}

\paragraph*{Inductive step.} 
By inductive hypothesis, 
$\InferStrong(M)$ terminates if and only if $\InferStrong(M_k)$ terminates.
Let $M_k=\ccontext[(\l x.P)Q] \redb \ccontext[P\subs{x}{Q}] = M_{k+1}$
following the $F_\infty$ reduction strategy, and let $\PDsnmin{M_k}=(\Pi_k, \eqset_k)$. 
Since $\Pi_k$ contains a subderivation with subject $(\lambda x.P)Q$, 
we can proceed as in the base step, 
proving that $\InferStrong(M_k)$ terminates if and only if $\InferStrong(M_{k+1})$ terminates. 
Therefore the property holds for all $k$.

\end{proof}

\begin{remark}
The proof technique for \Cref{lem:algo-as-redb}
can be used to address some gaps in a similar
proof of termination in \cite{PautassoR23}, 
which deals with $\kr$-reduction steps in an incorrect way.
\end{remark}